\documentclass[12pt,english,a4paper,twoside]{article}

\usepackage{graphicx}
\usepackage[hidelinks,linktocpage]{hyperref}  
\hypersetup{hypertexnames=false}
\usepackage{lmodern}
\usepackage[left=3cm,right=3cm,top=3cm,bottom=3cm]{geometry} 
\usepackage{url}                                 
\usepackage{booktabs}                      

\usepackage{amsthm}
\usepackage{amsmath}
\usepackage{amsfonts}
\usepackage{dsfont}                          
\usepackage{stix}                               
\usepackage{mathtools}
\usepackage{xparse, etoolbox}
\usepackage{enumerate}
\usepackage{mathrsfs}
\usepackage{array}
\usepackage{enumitem}
\usepackage{fancyhdr}                        
\usepackage{subdepth}                       
\usepackage{titlesec}                           
\usepackage{xypic}    
\usepackage{xfrac}   
\usepackage[numbers,sort]{natbib} 
\usepackage[ruled]{algorithm2e}
\SetKwInput{KwInput}{input}                
\SetKwInput{KwOutput}{output}          
\SetKwBlock{Repeat}{repeat}{} 
\usepackage[overload]{empheq}          
\usepackage[dvipsnames]{xcolor}
\usepackage[hang,flushmargin]{footmisc} 

\titleformat*{\paragraph}{\itshape}

\makeatletter
\let\amstexbig\big
\def\newbig#1{%
  \ifx#1|%
    \expandafter\@firstoftwo
  \else
    \expandafter\@secondoftwo
  \fi
  {\big@bar}%
  {\amstexbig{#1}}%
}
\AtBeginDocument{\let\big\newbig}
\def\big@bar{\bBigg@{1.1}|}
\makeatother

\def\SD{\mathscr{D}}
\def\SE{\mathscr{E}}

\def\SU{\mathscr{U}}

\def\CB{\mathcal{B}}

\def\CI{\mathcal{I}}

\def\CK{\mathcal{K}}

\def\CN{\mathcal{N}}

\def\BN{\mathbb{N}}

\def\BR{\mathbb{R}}

\def\BRnn{\BR_{\geqslant 0}}
\def\BRsp{\BR_{>0}}
\def\BNz{\BN_{0}}

\def\bfA{\mathbf{A}}
\def\bfB{\mathbf{B}}
\def\bfI{\mathbf{I}}
\def\bfH{\mathbf{H}}

\def\bfM{\mathbf{M}}

\def\bfQ{\mathbf{Q}}
\def\bfR{\mathbf{R}}
\def\bfS{\mathbf{S}}
\def\bfT{\mathbf{T}}

\def\bfx{\mathbf{x}}

\def\Theb{\boldsymbol{\Theta}}

\def\Sigb{\boldsymbol{\Sigma}}

\def\Omgb{\boldsymbol{\Omega}}

\def\bdd{\boldsymbol{d}}
\def\bde{\boldsymbol{e}}

\def\bdx{\boldsymbol{x}}

\DeclareMathOperator{\diag}{diag}
\DeclareMathOperator{\vspan}{span}

\DeclareMathOperator{\trace}{trace}
\DeclareMathOperator{\card}{card}

\def\act{\mathrm{supp}}
\def\edge{\mathrm{edge}}
\def\free{\mathrm{free}}

\newcommand{\score}{\mathsf{S}}

\def\Frob{\mathrm{F}}

\def\Mat{\mathcal{M}}
\def\Sym{\mathcal{S}}
\def\SPSD{\mathcal{S}_{\succcurlyeq 0}}
\def\SPD{\mathcal{S}_{\succ 0}}
\def\Diago{\mathcal{D}}

\DeclarePairedDelimiterX\mip[1]\langle\rangle{\mipargs{#1}}%
\NewDocumentCommand{\mipargs}{>{\SplitArgument{1}{,}}m }
 {\mipargsaux#1}
\NewDocumentCommand{\mipargsaux}{ m m }%
{\ifblank{#1}{\cdot}{#1}\nonscript\,\delimsize\vert\nonscript\,\mathopen{}\ifblank{#2}{\cdot}{#2}}

\theoremstyle{plain}
\newtheorem{theorem}{Theorem}[section]

\newtheorem{definition}{Definition}[section]
\newtheorem{lemma}{Lemma}[section]

\theoremstyle{definition}
\newtheorem{remark}{Remark}[section]

\newcommand{\fin} {\hfill\hbox{$\triangleleft$}}

\usepackage{tocloft}
\newcounter{myalgorithm}
\newenvironment{myalgo}[1]{%
    \refstepcounter{myalgorithm} 
    \begin{center}
    \noindent
    \begin{minipage}{\linewidth}
    \hrule
    \vspace{0.5ex}
    \textbf{Algorithm~\arabic{myalgorithm}:} \text{#1} 
    \vspace{0.5ex}
    \hrule
    \vspace{0.5ex}
}{%
    \vspace{0.5ex}
    \hrule
    \end{minipage}
    \end{center}
}

\title{\vspace*{-1cm} Information-geometry-driven\\ graph sequential growth}

\author{Harry~T. \textsc{Bond}\footnotemark[1]\qquad  
Bertrand \textsc{Gauthier}\footnotemark[1]\qquad
Kirstin \textsc{Strokorb}\footnotemark[2]}

\date{} 

\newcommand\shorttitle{Information-geometry-driven graph sequential growth}
\newcommand\authors{H.~T. \textsc{Bond}, B. \textsc{Gauthier} and K.  \textsc{Strokorb}}

\begin{document}

\maketitle

\renewcommand{\thefootnote}{\fnsymbol{footnote}}
%
\footnotetext[1]{Cardiff University, School of Mathematics --- {BondH@cardiff.ac.uk}, {GauthierB@cardiff.ac.uk}}
\footnotetext[2]{University of Bath, Department of Mathematical Sciences --- {KS3047@bath.ac.uk}}
\renewcommand{\thefootnote}{\arabic{footnote}}

\begin{abstract}
We investigate the properties of a class of regularisation-free approaches for Gaussian graphical inference
based on the information-geometry-driven sequential growth of initially edgeless graphs.  
Relating the growth of a graph to a coordinate descent process, 
we characterise the fully-corrective descents corresponding to information-optimal growths, 
and propose numerically efficient strategies for their approximation.   
We demonstrate the ability of the proposed procedures to reliably extract sparse graphical models 
while limiting the number of false detections, 
and illustrate how activation ranks can provide insight into the informational relevance of edge sets. 
The considered approaches are tuning-parameter-free and have complexities akin to coordinate descents.  
\end{abstract}


\vspace{0.25\baselineskip}

\tableofcontents

\section{Introduction}
\label{sec:Introduction}

Let $\bfx$ be a $d$-dimensional centred Gaussian random vector with covariance $\Sigb$ and precision $\Theb=\Sigb^{-1}$. 
 When the matrix $\Theb$ is sparse, the primary goal of graphical inference is to recover its sparsity pattern from the observation of independent realisations of $\bfx$. This task is often referred to as \emph{graph recovery} or \emph{structure learning} and is akin to learning the graph defined by the non-zero off-diagonal entries of $\Theb$; this graph describes the \emph{conditional independence structure} of $\bfx$; see e.g.\ \cite{lauritzen96, uhler2019handbook}.
Intrinsically related is the  complementary task of \emph{parameter learning}, that is, estimating the values of the non-null entries of $\Theb$, see for instance \cite{jog2015model}.
 More generally, practical motivation may stem from identifying  sparse approximations of not-necessarily-sparse precision matrices. 
There exists a rich literature on graphical model estimation and we refer the reader to \cite{drton2017structure,JankVdG2019handbook,uhler2019handbook,shutta2022gaussian} for overviews of major developments.
Established techniques include
partial-correlation testing, 
penalised likelihood maximisation, or neighbourhood selection 
\cite{drton2004model, meinshausen2006high, yuan2007model}. 

Likelihood-based approaches for Gaussian graphical inference 
relate to the minimisation of convex functions of the form 
\begin{align*}
f_{\bfS}:\bfQ\mapsto\trace(\bfS\bfQ) - \log(\det(\bfQ)), 
\quad\bfQ\in\SPD^{d}, 
\end{align*}
with $\SPD^{d}$ the convex cone of all order-$d$  symmetric positive-definite matrices, 
and where $\bfS$ is typically the sample covariance of the observations. 
We call $f_{\bfS}$  the \emph{Gaussian graphical loss} with respect to $\bfS$; see Section~\ref{sec:FrameAndNot} for details. 
Graphical inference is in this case carried out by searching for sparse approximations of the argument of the minimum of $f_{\bfS}$.  
In graphical-lasso-type approaches \cite{yuan2007model,banerjee2008model,friedman2008sparse}, this is achieved by minimising $f_{\bfS}$ in combination with a 
sparsity-inducing regularisation term.
A common approach is to minimise
\begin{align*}
\bfQ\mapsto f_{\bfS}(\bfQ)+\lambda\|\bfQ\|_{\ell^{1}}, 
\quad\bfQ\in\SPD^{d}, 
\end{align*}
with $\lambda\geqslant0$ and where $\|\bfQ\|_{\ell^{1}}=\sum_{i\neq j}|\bfQ_{i,j}|$ penalises the magnitude of the off-diagonal entries of $\bfQ$.   
Regularisation-based techniques enjoy very elegant properties and have been extensively studied in the literature; cf.~also~ \cite{Witten01012011,mazumder2012graphical,BuhlmVdG2011book,JankVdG2015EJS}.
In practice, they however require to appropriately tune the regularisation term, 
and they     can sometimes lead to relatively high ratios of false detections \cite{liu2013gaussian, laszkiewicz2021thresholded, koka2024false}; see also Section~\ref{sec:Experiments} for illustrations.

In this work, we investigate the properties of a class of regularisation-free  approaches for graphical inference based on the direct minimisation of $f_{\bfS}$ via information-geometry-driven coordinate descents with optimal diagonal initialisations. 
We leverage the analogy between edge activation and coordinate update, and characterise the descent corresponding to the likelihood-optimal growth of an initially edgeless graph. Specifically, optimal growths  correspond to \emph{fully-corrective} coordinate descents 
with \emph{best-fully-corrective-improvement} selection rule (see Section~\ref{sec:GraphGrow}). 
We then discuss numerically efficient strategies for the approximation of such descents, relying on approximate full correction
(Section~\ref{sec:CoordDesc})
and relaxed edge activation rules (Section~\ref{sec:RelaxedGrowth}).

The efficacy of the considered inference strategies is assessed on a comprehensive series of synthetic examples (Section~\ref{sec:Experiments}), thereby demonstrating the ability of 
the considered growth-based approaches to reliably extract sparse graphs  
while limiting the number of false detections. 
We also illustrate how the proposed strategies can 
assist the estimation of graphical models in practical situations. General comments and further considerations are gathered in a concluding discussion (Section~\ref{sec:Conclusion}).

From a theoretical standpoint, the presented developments bring together
classical results from probability, statistics and convex optimisation; 
to this extent, the contribution of this work is mostly
methodological, describing the theoretical foundations of the proposed approaches and demonstrating the practical interest of sequential-growth-based procedures as diagnostic tools for graphical-model estimation. For completeness, proofs of all the results stated in the main body of the paper are presented in Appendix~\ref{sec:ProofsAndCo}.

\section{Framework and notations}
\label{sec:FrameAndNot}

Throughout this note, we use the classical matrix notation 
and identify a vector $\bdx\in\BR^{d}$, $d\in\BN$,  
as the $d\times 1$ column matrix defined by the coefficients 
of $\bdx$ in the canonical basis $\{\bde_{i}\}_{i\in[d]}$ of $\BR^{d}$; 
$[d]$ stands for the set of all integers between $1$ and $d$.  
The transpose of a matrix $\bfM$ is denoted by $\bfM^{T}$,  
and $\bfI_{d}$ is the $d\times d$ identity matrix; 
$\|\bfM\|_{\Frob}$ stands for the Frobenius norm of $\bfM$, 
and the underlying inner product is denoted by $\mip[]{,}_{\Frob}$. 
 Let $\Mat^{d}=\BR^{d\times d}$ be the linear space of all $d$-order matrices; we consider the usual topology on $\Mat^{d}$.  
The linear subspaces of all diagonal and all symmetric matrices are denoted by $\Diago^{d}$ and $\Sym^{d}$, respectively.  
Further, $\SPSD^{d}$ and $\SPD^{d}$ stand for the convex cones of all symmetric positive-semidefinite and symmetric positive-definite matrices of order $d$, respectively.  
For $\CK\subseteq\SPD^{d}$, we set 
\begin{align*}
\lambda_{\min}(\CK)=\inf\{\lambda_{\min}(\bfQ) \,|\, \bfQ\in\CK\}
\quad\text{and}\quad
\lambda_{\max}(\CK)=\sup\{\lambda_{\max}(\bfQ) \,|\, \bfQ\in\CK\},
\end{align*}
where $\lambda_{\min}(\bfQ)$ and $\lambda_{\max}(\bfQ)$ are the smallest and largest eigenvalues of $\bfQ$, 
respectively. We also set $\BNz=\BN\cup\{0\}$.   

\subsection{Gaussian graphical loss}
\label{sec:GraphModInf}

Consider a $d$-dimensional centred Gaussian random vector  $\bfx\sim\CN_{\Sigb}$, with $\Sigb\in\SPD^{d}$. 
Let $(\bdx_{l})_{l\in[n]}$ be $n\in\BN$ independent realisations of $\bfx$. 
The negative log-likelihood of $\bfQ\in\SPD^{d}$ being the precision of $\bfx$ reads 
\begin{align}\label{eq:LikelihoodGraph}
-\ell(\CN_{\bfQ^{-1}}\,|\,\bdx_{1},\cdots,\bdx_{n}) = 
\frac{n}{2}\big[
\trace(\hat{\Sigb}\bfQ)
-\log(\det(\bfQ)) 
+d \log(2\pi) 
\big], 
\end{align}
with $\hat{\Sigb}=\frac{1}{n}\sum_{l=1}^{n}\bdx_{l}\bdx_{l}^{T}\in\SPSD^{d}$ the classical biased estimate of $\Sigb$. 
More generally, we call 
\begin{align*}
f_{\bfS}:\SPD^{d}\to\BR,\quad
f_{\bfS}(\bfQ) =\trace(\bfS\bfQ) - \log(\det(\bfQ))  
\end{align*}
the \emph{Gaussian graphical loss} with respect to $\bfS\in\SPSD^{d}$; see Remark~\ref{rem:AboutScoreKL} for an information geometric interpretation of the map $f_\bfS$. 

Observing that $\Sym^{d}$ is the tangent space of $\SPD^{d}$ at any given $\bfQ\in\SPD^{d}$, 
the first and second total derivatives of $f_{\bfS}$ at $\bfQ\in\SPD^{d}$ read
\begin{gather*}
D_{\bfQ} f_{\bfS} : \bfH \mapsto \trace((\bfS - \bfQ^{-1})\bfH)
=\mip[]{\bfS-\bfQ^{-1} , \bfH}_{\Frob}, \quad 
\bfH\in\Sym^{d},\mbox{ and} \\
D^2_{\bfQ} f_{\bfS} : (\bfH_{1},\bfH_{2}) 
	\mapsto  \trace(\bfQ^{-1} \bfH_{1} \bfQ^{-1} \bfH_{2}) 
	= \mip[]{\bfH_{1} \bfQ^{-1}, \bfQ^{-1} \bfH_{2}}_{\Frob}, \quad 
\bfH_{1}\text{ and }\bfH_{2}\in\Sym^{d}.   	
\end{gather*} 
In particular, $f_{\bfS}$ is continuous over $\SPD^{d}$. 
Lemma~\ref{lem:fSStrictCvx} below recalls some important properties of the map $f_{\bfS}$.  
We say that a function $f:\SPD^{d}\to\BR$ is \emph{coercive}\footnote{We use the term \emph{coercive} for convenience to encapsulate the idea of growing to $+\infty$ when reaching the boundary of the domain $\SPD^{d}$.}
if for any sequence 
$(\bfQ_{k})_{k\in\BN}\subset\SPD^{d}$ such that 
$\lambda_{\min}(\bfQ_{k})\to0$ or $\lambda_{\max}(\bfQ_{k})\to+\infty$, 
we have $f(\bfQ_{k})\to+\infty$. If $f$ is also continuous, 
then its sublevel sets are compact.

\begin{lemma}\label{lem:fSStrictCvx}
For all $\bfS\in\SPSD^{d}$,  the map $f_{\bfS}$ is strictly convex on $\SPD^{d}$. 
If $\bfS$ is invertible, then $f_{\bfS}$ is coercive over $\SPD^{d}$ 
and is minimised at $\bfQ=\bfS^{-1}$. 
If $\bfS$ is singular, then $f_{\bfS}$ is unbounded from below. 
\end{lemma}

Let $\nabla f_{\bfS}(\bfQ)=\bfS - \bfQ^{-1}$ be the gradient of $f_{\bfS}$ at $\bfQ\in\SPD^{d}$ with respect to the Frobenius inner product.
Lemma~\ref{lem:StongConvexity} expresses that
the map $f_{\bfS}$ is strongly convex and has Lipschitz-continuous gradient on any non-empty compact convex subset of $\SPD^{d}$. 

\begin{lemma}\label{lem:StongConvexity}
Let $\bfS\in\SPSD^{d}$.
If $\CK\subset\SPD^{d}$ is such that
$\lambda_{\min}(\CK)>0$, then we have  
\begin{align*}
\|\nabla f_{\bfS}(\bfQ_{1})-\nabla f_{\bfS}(\bfQ_{2})\|_{\Frob}
\leqslant\|\bfQ_{1}-\bfQ_{2}\|_{\Frob}/\lambda_{\min}^{2}(\CK), 
\quad \text{$\bfQ_{1}$ and $\bfQ_{2}\in\CK$.} 
\end{align*}
If $\CK\subset\SPD^{d}$ is such that 
$\lambda_{\max}(\CK)$ is finite, then we have 
\begin{align*}
\mip[]{\nabla f_{\bfS}(\bfQ_{1})-\nabla f_{\bfS}(\bfQ_{2}), \bfQ_{1}-\bfQ_{2} }_{\Frob}
\geqslant\|\bfQ_{1}-\bfQ_{2}\|_{\Frob}^{2}/\lambda_{\max}^{2}(\CK), 
\quad \text{$\bfQ_{1}$ and $\bfQ_{2}\in\CK$.} 
\end{align*}
\end{lemma}

In view of Lemma~\ref{lem:fSStrictCvx}, 
the negative log-likelihood \eqref{eq:LikelihoodGraph} admits a minimum over $\SPD^{d}$
solely when $\hat{\Sigb}$ is invertible. 
In situations where $\hat{\Sigb}$ is singular (for $n\leqslant d$ for instance), 
it is common to introduce a \emph{ridge penalty} and consider the graphical loss defined by
$\bfS=\hat{\Sigb}+\bfT$, with $\bfT\in\SPD^{d}$.  
A typical instance is to use $\bfT=\gamma^{2}\bfI_{d}$, with $\gamma>0$ (Gaussian observation noise); 
see for instance \cite{ha2014partial, van2020updating}. 
A ridge penalty can also help to improve numerical stability or to account for account measurement errors.

\begin{remark}\label{rem:AboutScoreKL}
The loss function $f_{\bfS}$ is closely related to the 
information geometry of multivariate Gaussian distributions, which dates back to the works of
\cite{kullback97reprint, csiszar75idiv, speedkiiveri86}.
Following for instance \cite{dawid07geometryscoring},
the Gaussian graphical loss 
can be recovered from the \emph{scoring rule} 
\begin{align*}
\score(\bdx, &\CN_{\bfQ^{-1}}) = \bdx^{T}\bfQ\bdx - \log(\det(\bfQ)), \quad 
\bdx\in\BR^{d}\text{ and }\bfQ\in\SPD^{d}.
\end{align*}
 When $\bfS$ is invertible, we have
\begin{align*}
 f_{\bfS}(\bfQ)-f_{\bfS}(\bfS^{-1})
 =\trace(\bfS\bfQ -\bfI_{d}) - \log(\det(\bfS\bfQ))
 =2 D_{\mathrm{KL}}(\CN_{\bfS} \| \CN_{\bfQ^{-1}}),
 \end{align*}
 with $D_{\mathrm{KL}}(\CN_{\bfS} \| \CN_{\bfQ^{-1}})$ the \emph{Kullback-Leibler divergence} from $\CN_{\bfQ^{-1}}$ to $\CN_{\bfS}$.  
 Here, the value 
$f_{\bfS}(\bfS^{-1})=d+\log(\det(\bfS))$ 
is the \emph{generalised entropy} or \emph{information measure} of the distribution $\CN_{\bfS}$.
\fin
\end{remark}

\subsection{Coordinate directions and graph likelihood}
\label{sec:CoordDir}

For $i$ and $j\in[d]$, $i<j$, let 
\begin{align*}
\bfB{(i,i)}=\bde_{i}\bde_{i}^{T} 
\quad\text{and}\quad
\bfB{(i,j)}  = [\bde_{i}\bde_{j}^{T} + \bde_{j}\bde_{i}^{T} ]/\sqrt{2},
\end{align*}
so that the collection $\CB=\{\bfB{(i,j)} \,|\, 1 \leqslant i < j \leqslant d\}$ 
forms an orthonormal basis of $\Sym^{d}$ for the Frobenius inner product.
When minimising $f_{\bfS}$, 
a line search along $\bfB{(i,i)}$ updates the $i$-th diagonal entry of the current iterate (\emph{diagonal update}), 
while a line search along $\bfB{(i,j)}$ updates the $(i,j)$-entry and its symmetric counterpart (\emph{off-diagonal update}).  
We denote by 
\begin{align*}
\SD=\{(i,i)\}_{i\in[d]}
\quad\text{and}\quad \SU=\{(i,j)\}_{1\leqslant i<j\leqslant d} 
\end{align*}
the sets of all diagonal and upper-diagonal indices of a $d\times d$ matrix, respectively.

We use the following terminology to make 
the connections between graph sequential growth strategies and associated coordinate descents precise.
For $\bfQ\in\SPD^{d}$, the \emph{support} (or \emph{activation}) $\act(\bfQ)$ of $\bfQ$ is given by
\begin{align*}
\act(\bfQ)=\big\{(i,j)\in\SD\cup\SU \, \big| \, \bfQ_{i,j}\neq0\big\}.  
\end{align*}
Observe that 
we always have $\SD\subseteq\act(\bfQ)$.
The set of upper-diagonal indices
We refer to
\[\edge(\bfQ)=\act(\bfQ)\cap\SU
\quad \text{and} \quad 
\free(\bfQ)=\SU\backslash\edge(\bfQ),
\] 
as the \emph{edge set defined by $\bfQ$} and the \emph{free set of $\bfQ$}.  Naturally, $\edge(\bfQ)$ defines the edge set of a simple $d$-vertex undirected graph. We say that \emph{a matrix $\bfQ\in\SPD^{d}$ is supported by a graph $G=([d],\SE)$} if $\edge(\bfQ)\subseteq\SE$, where $\SE\subseteq\SU$ is the edge set of $G$. 

For $\bfS\in\SPSD^{d}$ and $G=([d],\SE)$, $\SE\subseteq\SU$, 
the \emph{Gaussian graphical loss $\widetilde{f}_{\bfS}(G)$ of the graph $G$ with respect to $\bfS$} is 
the infimum of $f_{\bfS}$ over the set of all symmetric positive definite matrices supported by $G$, that is,    
\begin{align}
\begin{split}\label{def:LikeliGraph}
\widetilde{f}_{\bfS}(G) 
& =\inf\big\{f_{\bfS}(\bfQ) \, \big| \,   \bfQ \in \SPD^{d}, \, \edge(\bfQ)\subseteq\SE \big\} \\
& = \inf\big\{f_{\bfS}(\bfQ) \, \big| \, \bfQ\in\vspan\{\bfB(i,j)|(i,j)\in\SD\cup\SE\}\cap\SPD^{d}\big\}.  
\end{split}
\end{align}
For $\bfS\in\SPSD^{d}$ such that $\diag(\bfS)\in\BRsp^{d}$, 
the graphical loss of the  
edgeless graph
${G_{\emptyset}=([d], \emptyset)}$
is realised by
the diagonal matrix $\bfQ_{\bfS,\emptyset}\in\Diago^{d}\cap\SPD^{d}$, with 
\begin{align} 
\label{eq:opt-diag}
[\bfQ_{\bfS,\emptyset}]_{i,i}=1/\bfS_{i,i}, \quad i\in[d].\end{align}
Indeed, we have $[\nabla f_{\bfS}( \bfQ_{\bfS,\emptyset})]_{i,i}=0$, $i\in [d]$, 
and so $\widetilde{f}_{\bfS}(G_{\emptyset})=f_{\bfS}(\bfQ_{\bfS,\emptyset})$.  
More generally, when $\bfS$ is invertible, the infimum appearing in \eqref{def:LikeliGraph} is always a minimum.

\begin{lemma}
\label{lem:GraphLikeliExist}
Let $\bfS\in\SPD^{d}$. For $G=([d],\SE)$, $\SE\subseteq\SU$, 
there exists a unique matrix $\bfQ_{\bfS,\SE}\in\SPD^{d}$ such that 
$\edge(\bfQ)\subseteq\SE$ and 
$\widetilde{f}_{\bfS}(G)=f_{\bfS}(\bfQ_{\bfS,\SE})$. 
We have $[\nabla f_{\bfS}(\bfQ_{\bfS,\SE})]_{i,j}=0$ for all $(i,j)\in\SD\cup\SE$.  
\end{lemma}

For a general graph $G=([d],\SE)$, an analytical expression for the minimising matrix $\bfQ_{\bfS,\SE}$ does not exist.    
The properties of such matrices, and algorithms for their approximation, 
are for instance presented in \cite{uhler2019handbook}, an early key reference on relevant matrix completions being \cite{speedkiiveri86};   see also Section~\ref{sec:CoordDesc}. 
We refer to $\bfQ_{\bfS,\SE}$ as the \emph{graph-optimal matrix for $\bfS$ and $G$}.

\section{Information-optimal graph sequential growth}
\label{sec:GraphGrow}

We use the term \emph{graph sequential growth} for the process of sequentially adding edges to an initially edgeless graph, 
the vertex set $[d]$ staying fixed.   
Since a sequential growth corresponds to the definition of a \emph{total order} on the set of upper-diagonal indices $\SU$, 
setting $M=\card(\SU)=d(d-1)/2$, 
we can equivalently identify a graph sequential growth with  the specification of 
a finite sequence $\{(i,j)^{(k)}\}_{k=1}^{M}$ of distinct edges. 
We denote the underlying graphs as
\begin{align*} 
G^{(0)}=([d],\emptyset), 
\quad\text{and}\quad
G^{(k)}=([d],\SE^{(k)}), \quad 
\text{with} \quad  \SE^{(k)}=\{(i,j)^{(l)}\}_{l=1}^{k}, k\in[M].    
\end{align*}
Observe that $\card(\SE^{(k)})=k$ and $\SE^{(M)}=\SU$. 
Graph-sequential growths corresponds to forward-search methods, see for instance \cite{banerjee2008model}.
We refer to the operation of adding an edge to a current graph as \emph{edge activation}.

\begin{definition}
A sequence of distinct edges $\{(i,j)^{(k)}\}_{k=1}^{M}$ defines
an information-optimal graph sequential growth for $\bfS\in\SPD^{d}$ 
if for all $k\in[M]$, we have
\begin{align*}
\widetilde{f}_{\bfS}(G^{(k)}) = \min\big\{\widetilde{f}_{\bfS}(G) \, \big| \,  G=([d], \SE^{(k-1)}\cup\{(i,j)\}), \,
 (i,j)\in\SU\backslash\SE^{(k-1)} \big\}. 
\end{align*}
\end{definition}

If $\{(i,j)^{(k)}\}_{k=1}^{M}$ defines
an information-optimal graph sequential growth for $\bfS\in\SPD^{d}$, 
then the underlying sequence of graph-optimal matrices $\{\bfQ_{\bfS,\SE^{(k)}}\}_{k=0}^{M}$ corresponds to the iterates 
of a coordinate descent for the minimisation of $f_{\bfS}$, 
with \emph{optimal diagonal initialisation}, 
\emph{fully-corrective update rule}, and
\emph{best-fully-corrective-improvement (BFCI) selection rule}; as detailed  in Algorithm~\ref{algo:SeqGrowth} and the subsequent discussion. 
 
For $\bfQ\in\SPD^{d}$ and $(i,j)\in\SU$, define the \emph{fully-corrective improvement score}
\begin{align*}
\CI_{\bfS}^{\mathrm{FC}}\big(\bfQ, (i,j)\big)
=f_{\bfS}(\bfQ) - \widetilde{f}_{\bfS}(G), 
\quad \text{with} \quad
G=\big([d],\edge(\bfQ)\cup\{(i,j)\}\big).  
\end{align*}
Note that 
for $G=([d],\SE)$ and $(i,j)\in\SE$, we have 
$\CI_{\bfS}^{\mathrm{FC}}\big(\bfQ_{\bfS,\SE}, (i,j)\big)=0$.

\begin{myalgo}{Information-optimal graph sequential growth. }
\label{algo:SeqGrowth}
\textbf{given} $\bfS\in\SPD^{d}$. \\ 
Set $\bfQ^{(0)}=\bfQ_{\bfS,\emptyset}$ and $\SE^{(0)}=\emptyset$ (initialisation, see \eqref{eq:opt-diag} above).  \\
\textbf{for $k\in[M]$ do} 
\begin{enumerate}[nosep,leftmargin=2em]
\item Compute $(i,j)^{(k)}\in\arg\max\big\{ \CI_{\bfS}^{\mathrm{FC}}\big(\bfQ^{(k-1)}, (i,j)\big) \, \big| \, (i,j)\in\free(\bfQ^{(k-1)}) \big\}$. 
\item Set $\SE^{(k)}=\SE^{(k-1)}\cup\{(i,j)^{(k)}\}$ and $\bfQ^{(k)}=\bfQ_{\bfS,\SE^{(k)}}$. 
\end{enumerate}
\textbf{return} $\{(i,j)^{(k)}\}_{k\in[M]}$. 
\end{myalgo}

The interpretation of Algorithm~\ref{algo:SeqGrowth} as a coordinate descent with fully-corrective update 
follows from observing that for $k\in[M]$, we have
\begin{align*}
f_{\bfS}(\bfQ^{(k)})=\min\Big\{f_{\bfS}(\bfQ) \, \Big| \,
	\bfQ\in\vspan\big\{\bfB(i,j) \, \big| \, (i,j)\in\act(\bfQ^{(k-1)})\cup\{(i,j)^{(k)}\}\big\}\Big\}.   
\end{align*}
In addition to $\bfB(i,j)^{(k)}$ and the diagonal coordinate directions $\bfB(i,i)$, $i\in[d]$, 
for $k\geqslant 2$, the computation of the iterate $\bfQ^{(k)}$
therefore also involves all the previously activated off-diagonal directions $\bfB(i,j)^{(l)}$, $l\in[k-1]$. 
The selection rule in Algorithm~\ref{algo:SeqGrowth} 
amounts to selecting the edge leading to the 
best fully-corrective improvement, that is, 
the maximum improvement after full correction.

As mentioned in Section~\ref{sec:CoordDir}, for a general graph $G=([d],\SE)$, 
it is generally not possible to obtain precisely the graph-optimal matrix $\bfQ_{\bfS,\SE}$ 
(let alone in exact arithmetic), 
and one needs instead to rely on iterative approximation procedures; see Section~\ref{sec:CoordDesc}. 
Moreover, for a current iterate $\bfQ^{(k-1)}$, 
the characterisation of the optimal edge $(i,j)^{(k)}$ involves the computation 
of $M-k+1$ such matrices. 
From a numerical standpoint, the 
execution
of Algorithm~\ref{algo:SeqGrowth} 
is thus computationally prohibitive,
especially when the dimension $d$ is large; this limitation  is generally inherent to forward-backward searches
\cite{uhler2019handbook,banerjee2008model}. 
A natural alternative 
is to consider
coordinate descents with relaxed selection and update rules
as discussed in Section~\ref{sec:RelaxedGrowth}. 
In practice, one might not necessarily compute a sequential growth in full, 
and instead stop the growth procedure once a maximum number of edges $k_{\max}$ is reached (partial growth).

\section{Coordinate descent and graph optimality}
\label{sec:CoordDesc}

In this section, we discuss the minimisation of $f_{\bfS}$, $\bfS\in\SPD^{d}$, over the convex cone of all symmetric positive-definite matrices supported by a given graph, that is,  the approximation of matrices of the form $\bfQ_{\bfS,\SE}$, $\SE\subseteq\SU$.
More specifically, we consider exact group-coordinate descents that consist of either updating a diagonal entry of the current iterate, or one of its $2\times2$ principal submatrices; we refer to such updates as \emph{$\{1,2\}$-updates} (see Section~\ref{sec:BlockUpdate}).

\subsection{Exact block update}
\label{sec:BlockUpdate}

The Gaussian loss $f_{\bfS}$ can be minimised via exact coordinate descent, 
that is, by performing exact line searches along directions in the set of coordinate directions $\CB$ (see Section~\ref{sec:CoordDir}). 
It is nevertheless also possible to characterise exact group-coordinate updates 
by directly updating a principal submatrix of the current iterate, 
as described in Lemma~\ref{lem:BlockUpdate} 
(see for instance \cite[Lemma 2]{speedkiiveri86}, and also \cite{hojsLaur24, lauritzen96}).  

\begin{lemma}[exact order-$m$ update]\label{lem:BlockUpdate} 
Consider $\bfQ \in \SPD^{d}$ and 
let $I\subseteq[d]$ be a subset of size $m$; 
set $\bfR=\bfQ^{-1}$.  
The minimum of $f_{\bfS}$ over 
$(\bfQ + \vspan\{\bfB(i,j) \,|\, \text{$i, j\in I$, $i\leqslant j$}\})\cap \SPD^{d}$ is reached at 
$\tilde{\bfQ}\in\SPD^{d}$, with
\begin{align*}
\tilde{\bfQ}_{I,I} = \bfQ_{I,I} + (\bfS_{I,I})^{-1} - (\bfR_{I,I})^{-1}
\quad\text{and}\quad\text{$\tilde{\bfQ}_{i,j} = \bfQ_{i,j}$ \quad for $i$ or $j\not\in I$. } 
\end{align*}
Setting $\tilde{\bfR}=\tilde{\bfQ}\vphantom{\bfQ}^{-1}$, we then have 
\begin{align*}
\begin{pmatrix}
   \tilde{\bfR}_{I,I} &    \tilde{\bfR}_{I,I^{c}} \\
    \tilde{\bfR}_{I^{c},I} &    \tilde{\bfR}_{I^{c},I^{c}} 
\end{pmatrix} =     
\begin{pmatrix}
    \bfS_{I,I} &    \bfS_{I,I}  (\bfR_{I,I})^{-1} \bfR_{I,I^{c}} \\ 
    \bfR_{I^{c},I}  (\bfR_{I,I})^{-1} \bfS_{I,I}  
                & \bfR_{I^{c},I^{c}} - 
                    \bfR_{I^{c},I}(\bfR_{I,I})^{-1}[\bfI_{m} - \bfS_{I,I} (\bfR_{I,I})^{-1}] \bfR_{I,I^{c}}. 
\end{pmatrix}, 
\end{align*}
and  
$f_{\bfS}(\bfQ) - f_{\bfS}(\tilde{\bfQ}) 
=\trace(\bfS_{I,I} (\bfR_{I,I})^{-1}) - m - \log(\det(\bfS_{I,I}(\bfR_{I,I})^{-1}))$. 
\end{lemma}

In Lemma~\ref{lem:BlockUpdate}, the case $I=\{i\}$, $i\in[d]$, 
describes an exact line search along $\bfB(i,i)$; 
we refer to such an update as a \emph{$1$-update} (or \emph{order-$1$ update}).  
The case $I=\{i,j\}$, $(i,j)\in\SU$, corresponds to the simultaneous update of $\bfQ_{i,j}$, $\bfQ_{j,i}$, $\bfQ_{i,i}$ and $\bfQ_{j,j}$; 
we refer to such an update as a \emph{$2$-update} (or \emph{order-$2$ update}).  
Higher order updates may also be considered; 
such updates relate to the existence of cliques in the underlying graphs, 
see for instance \cite{speedkiiveri86}. 
In what follows, we only consider \emph{$\{1,2\}$-updates}, that is, order-$1$ or order-$2$ updates.

\subsection{Graph-optimal matrix approximation}
\label{sec:GraphOptApprox}

For $\bfS\in\SPD^{d}$ and $G=([d],\SE)$, the following general procedure (Algorithm~\ref{algo:GraphOptMatApp12Up}) 
produces sequences of iterates converging towards $\bfQ_{\bfS,\SE}$. 
In the framework of Algorithm~\ref{algo:SeqGrowth}, Algorithm~\ref{algo:GraphOptMatApp12Up} can be used to approximate the full correction related to the activation of a given edge in the free set of a current iterate. 
In practice, the optimisation is carried out until a given stopping criterion is satisfied 
(see Remark~\ref{rem:StopAlgo2} below); 
for numerical efficiency, rather than performing matrix inversion, one should use Lemma~\ref{lem:BlockUpdate} to update the inverse of the current iterates. 

\begin{myalgo}{
Coordinate descent with $\{1,2\}$-updates and support restriction. }
\label{algo:GraphOptMatApp12Up}
\textbf{given} $\bfS\in\SPD^{d}$, $G=([d],\SE)$, and $\bfQ^{(0)}\in\SPD^{d}$ such that $\edge(\bfQ^{(0)})\subseteq\SE$. \\
Set $t=1$. \\
\textbf{repeat} 
\begin{enumerate}[nosep,leftmargin=2em]
\item Select $(i,j)^{(t)}\in\SD\cup\SE$ (selection rules are discussed below). 
\item Perform the corresponding $\{1,2\}$-update to define $\bfQ^{(t)}$ from $\bfQ^{(t-1)}$\\
(use Lemma~\ref{lem:BlockUpdate} with $\bfQ=\bfQ^{(t-1)}$, 
$\bfR=(\bfQ^{(t-1)})^{-1}$
and $\tilde{\bfQ}=\bfQ^{(t)}$); 
increment $t$. 
\end{enumerate}
\textbf{return} $\{\bfQ^{(t)}\}_{t\in\BNz}$.
\end{myalgo}

By construction, and independently of the selection rule considered, the sequence $\{f_{\bfS}(\bfQ^{(t)})\}_{t\in\BNz}$ is non-increasing and for all $t\in\BNz$ we have 
\begin{align*}
\bfQ^{(t)}\in
\big\{\bfQ\in\SPD^{d} \, \big|\, f_{\bfS}(\bfQ) \leqslant f_{\bfS}(\bfQ^{(0)})
\text{ and }\edge(\bfQ)\subseteq\SE\big\}
= \CK_{0,\SE}.
\end{align*}
By Lemma~\ref{lem:fSStrictCvx}, $\CK_{0,\SE}$ is convex and compact 
(as it is the intersection of a convex compact sublevel set of $f_{\bfS}$ with a linear subpace of $\Sym^{d}$). 
The strong-convexity of $f_{\bfS}$ and the Lipschitz-continuity of its gradient over $\CK_{0,\SE}$ (see Lemma~\ref{lem:StongConvexity}) thus ensure the convergence of Algorithm~\ref{algo:GraphOptMatApp12Up} for a variety of selection rules, 
including the cyclic rule, the random-uniform rule, or various gradient-based rules;  
see for instance \cite{nesterov2012efficiency, bubeck2015convex, nutini2015coordinate, karimi2016linear} and references therein. The convergence is generally at-worst-linear, 
as 
shown
below for the \emph{Gauss-Southwell} (GS) selection rule, which consists 
of considering the coordinate directions 
\begin{align*}
(i,j)^{(t)}\in\arg\max_{(i,j)\in\SD\cup\SE} |\mip[]{\nabla f_{\bfS}(\bfQ^{(t-1)}), {\bfB(i,j)}}_{\Frob}|, 
\quad t\in\BN.   
\end{align*}

\begin{theorem}\label{thm:ConvergenceAlgo2GS}
For the GS selection rule, the sequence of iterates $\{\bfQ^{(t)}\}_{t\in\BNz}$ generated by Algorithm~\ref{algo:GraphOptMatApp12Up} satisfies
\begin{align*}
f_{\bfS}(\bfQ^{(t)}) - f_{\bfS}(\bfQ_{\bfS,\SE})
\leqslant \Big(1 - \frac{\mu}{mL}\Big)^{t}\big(f_{\bfS}(\bfQ^{(0)}) - f_{\bfS}(\bfQ_{\bfS,\SE})\big), \quad
t\in\BNz,    
\end{align*}
with $\mu=1/\lambda_{\max}^{2}(\CK_{0,\SE})$, 
$L=1/\lambda_{\min}^{2}(\CK_{0,\SE})$ 
and $m=\card(\SD\cup\SE)$. 
\end{theorem}

\begin{remark}\label{rem:StopAlgo2}
In the framework of Theorem~\ref{thm:ConvergenceAlgo2GS}, 
we have 
(see \eqref{eq:GapBoundImprov} in Appendix~\ref{sec:ProofsAndCo})
\begin{align*}
f_{\bfS}(\bfQ^{(t-1)}) - f_{\bfS}(\bfQ_{\bfS,\SE})
\leqslant \tfrac{mL}{\mu}\big(f_{\bfS}(\bfQ^{(t-1)}) - f_{\bfS}(\bfQ^{(t)})\big), \quad
t\in\BN,
\end{align*}  
which provides an (implicit) 
upper bound on the optimality gap based on the magnitude of the current improvement. 
 Hence, one may decide to stop the descent 
once the current improvement reaches a given fraction of the initial one, that is, once the condition 
\begin{align*}
f_{\bfS}(\bfQ^{(t-1)}) - f_{\bfS}(\bfQ^{(t)})\leqslant \tau \big(f_{\bfS}(\bfQ^{(0)}) - f_{\bfS}(\bfQ^{(1)})\big)
\end{align*}
is satisfied, with $0<\tau\leqslant 1$. 
In our experiments of Section~\ref{sec:Experiments}, we implement such a stopping rule 
in combination with an upper bound on the number of iterations 
of the form $\lceil\alpha m + \beta\rceil$, 
with $\alpha$ and $\beta\geqslant 0$; recall that $m=\card(\SE)$.
\fin
\end{remark}

\section{Relaxed graph sequential growth}
\label{sec:RelaxedGrowth}

The information-optimal graph sequential growth for $\bfS\in\SPD^{d}$ corresponds to a coordinate descent for the minimisation of $f_{\bfS}$ with fully-corrective update and BFCI selection rule (see Section~\ref{sec:GraphGrow}). 
To mitigate the numerical cost inherent to such descents, 
relaxed selection and update rules may be considered.
We focus on the ones below. Contrary to the BFCI rule, the presented rules do not involve the computation of graph-optimal matrices. Many variants of the described strategies could be considered.

\paragraph{Best block improvement (BBI).} Following Lemma~\ref{lem:BlockUpdate}, 
for $\bfQ\in\SPD^{d}$ and $(i,j)\in\SU$, 
setting $\bfR=\bfQ^{-1}$ and $I=\{i,j\}$, 
we define the \emph{block improvement score}
\begin{align*}
\CI_{\bfS}^{\mathrm{B}}\big(\bfQ, (i,j)\big)
=\trace(\bfS_{I,I} (\bfR_{I,I})^{-1}) - 2 - \log(\det(\bfS_{I,I}(\bfR_{I,I})^{-1})). 
\end{align*}
The BBI selection rule consists of selecting the 2-update leading to the best improvement, without accounting for full correction,
that is,
\begin{align*}
(i,j)^{(k)}\in\arg\max_{(i,j)\in\free(\bfQ^{(k-1)})} \CI_{\bfS}^{\mathrm{B}}\big(\bfQ^{(k-1)}, (i,j)\big).  
\end{align*}
 
\paragraph{Gauss-Southwell-Lipschitz (GSL).}  
The GSL rule is a variant of the GS rule that also accounts for the local Lipschitz continuity of the gradient, see \cite{nutini2015coordinate}; 
it is defined as 
\begin{align*}
(i,j)^{(k)}\in\arg\max_{(i,j)\in\free(\bfQ^{(k-1)})} \big(D_{\bfQ^{(k-1)}}f_{\bfS}(\bfB(i,j))\big)^{2}/D^{2}_{\bfQ^{(k-1)}}f_{\bfS}(\bfB(i,j),\bfB(i,j)).   
\end{align*}
Observe that $D^{2}_{\bfQ}f_{\bfS}(\bfB(i,j),\bfB(i,j))=\bfR_{i,i}\bfR_{j,j}+\bfR_{i,j}^{2}$, 
with $\bfR=\bfQ^{-1}$ and $(i,j)\in\SU$. 
In terms of the amount of work, the GSL rule is less costly than the BBI rule.  
Note that the GSL and BBI selection rules can also be used in the framework of Algorithm~\ref{algo:GraphOptMatApp12Up}. 
\bigskip

Algorithm~\ref{algo:RelaxSeqGrowth} below describes a general class of coordinate-descent-type strategies for graph sequential growth. 
The algorithm consists of two nested loops: 
the outer loop performs edge activation, 
and the inner loop uses Algorithm~\ref{algo:GraphOptMatApp12Up} 
to approximate full correction. 
Similarly to Algorithm~\ref{algo:SeqGrowth}, one might stop the procedure in practice once a maximum number of edges $k_{\max}$ is reached 
(partial growth).

\begin{myalgo}{Relaxed graph sequential growth. }
\label{algo:RelaxSeqGrowth}
\textbf{given} $\bfS\in\SPD^{d}$. \\ 
Set $\bfQ^{(0)}=\bfQ_{\bfS,\emptyset}$ and $\SE^{(0)}=\emptyset$ (initialisation).  \\
\textbf{for $k\in[M]$ do} 
\begin{enumerate}[nosep,leftmargin=2em]
\item Select $(i,j)^{(k)}\in\free(\bfQ^{(k-1)})$ 
(edge activation, see selection rules above),  \\ 
and set $\SE^{(k)}=\SE^{(k-1)}\cup\{(i,j)^{(k)}\}$.  
\item Use Algorithm~\ref{algo:GraphOptMatApp12Up}
initialised at $\bfQ^{(k-1)}$
to compute an approximation $\bfQ^{(k)}$ of $\bfQ_{\bfS,\SE^{(k)}}$\\ 
(approximate full correction, see the stopping rule in Remark~\ref{rem:StopAlgo2}). 
\end{enumerate}
\textbf{return} $\{(i,j)^{(k)}\}_{k\in[M]}$. 
\end{myalgo}

In our experiments of Section~\ref{sec:Experiments}, we also implement a variant of Algorithm~\ref{algo:RelaxSeqGrowth} 
relying on approximate full corrections to approximate the BFCI selection rule. More precisely, we systematically use Algorithm~2  
with the stopping rule discussed in Remark~\ref{rem:StopAlgo2} to approximate the required graph-optimal matrices. 
In terms of amount of work and for identical stopping rules for approximate full correction, 
computing the full BFCI growth 
requires $M(M+1)/2$ full corrections, whereas 
the full GSL or BBI growths involve only $M$ full corrections.

\section{Experiments}
\label{sec:Experiments}

We illustrate  the graph-recovery ability 
of the discussed strategies (Algorithm~\ref{algo:RelaxSeqGrowth}) on a series of examples. 
In Sections~\ref{sec:494bus} and \ref{sec:FullSynthe}, 
we consider synthetic data
and compare the proposed graph-sequential-growth procedures with graphical lasso (Glasso; see Remark~\ref{rem:AboutGlasso}), 
and with the naive sequential growths defined by the magnitude of the entries of the sample precision and partial-correlation matrices
(see Remark~\ref{rem:NaiveMagnit}). 
An application involving gene-expression data is presented in Section~\ref{sec:Riboflavin}.

To ensure invertibility and for simplicity,  
we systematically add a ridge penalty to the sample covariance matrix and work with $\bfS=\hat{\Sigb}+\gamma^{2}\bfI_{d}$, with $\gamma^{2}=\rho\sum_{i\in[d]}\hat{\Sigb}_{i,i}/d$ and $\rho=10^{-6}$. 
To approximate full correction, we use Algorithm~\ref{algo:GraphOptMatApp12Up} with GSL selection rule and implement the stopping rule from Remark~\ref{rem:StopAlgo2} with, unless otherwise stated, $\alpha=1$, $\beta=10$ and $\tau=10^{-5}$. 
We deliberately consider matrices of moderate dimension; this allows us to explore the regularisation paths of Glasso, 
compute sequential growths in full, and perform multiple repetitions of the considered experiments.

\begin{remark}\label{rem:AboutGlasso}
We consider the standard Glasso formulation (see Section~\ref{sec:Introduction}) 
and use the \texttt{GGlasso} Python package  \cite{Schaipp2021} for computations. 
For a given matrix $\bfS$, we approximate the regularisation path of Glasso by discretising 
the interval $[0,\max_{i\neq j}|\bfS_{i,j}|]$; 
we use $100$ uniformly-spaced values of the regularisation parameter. 
\fin
\end{remark}

\begin{remark}[naive graph sequential growths]\label{rem:NaiveMagnit}
For $\bfS\in\SPD^{d}$, we refer to
the graph sequential growth induced by the ordering, 
in decreasing order,  
of the absolute values of the upper-diagonal entries of $\Omgb=\bfS^{-1}$
as the \emph{precision-magnitude growth}; ties are broken randomly. 
Similarly, we define the \emph{partial-correlation-magnitude growth} as the 
growth induced by the ordering of the absolute values of the upper-diagonal entries of 
the matrix with $i,j$ entry $-\Omgb_{i,j}/(\Omgb_{i,i}\Omgb_{j,j})^{1/2}$. 
For short, the resulting growth procedures are referred to as {Prec} and {PCorr}, respectively. 
These naive growths are implemented for comparison purposes; they relate to thresholding-based approaches. \fin
\end{remark}

\subsection{Synthetic example based on \texttt{494\_bus}}
\label{sec:494bus}

We consider the symmetric positive-definite matrix $\bfM_{0}$ defined by the \texttt{494\_bus} instance
of the SuiteSparse \mbox{Matrix} Collection \cite{davis2011university}, where we
 extract the second $50\times 50$ diagonal block of $\bfM_{0}$, and denote the resulting matrix by $\bfM$. 
We normalise the matrix $\bfM^{-1}$ (covariance to correlation), 
and use the resulting matrix $\Sigb$ as true covariance. 
We generate random samples from $\CN_{\Sigb}$ of size $n$, and use these samples to perform graphical-model estimation. 
The sparsity pattern of $\Theb=\Sigb^{-1}$ is presented in Figure~\ref{fig:LabelFig1}.
The underlying graph has $d=50$ vertices and $30$ edges,
and the values of the corresponding precisions and partial correlations are given in Figure~\ref{fig:LabelFig3}. 
Note that all the partial correlations are non-negative. 
   
We first compare the accuracy of the edge sets selected by Algorithm~\ref{algo:RelaxSeqGrowth} with GSL, BBI and BFCI selection rules. We consider the sample sizes $n=30$, $90$ and $160$, and in each case perform $100$ repetitions for GSL and BBI (that is, we generate 100 random samples from $\CN_{\Sigb}$),
and $10$ repetitions for BFCI.
Figure~\ref{fig:LabelFig1} shows the resulting graph-recovery ROC and precision-recall curves, the latter placing more emphasis on the accuracies achieved in the early stages of the extraction procedures; FP stands for false positive.
We do not observe any notable differences between the considered selection rules, supporting the use of  GSL or BBI selection rules as numerically efficient approximations of the BFCI rule.  

\begin{figure}[!htbp]
\centering
\includegraphics[width=1.0\linewidth]{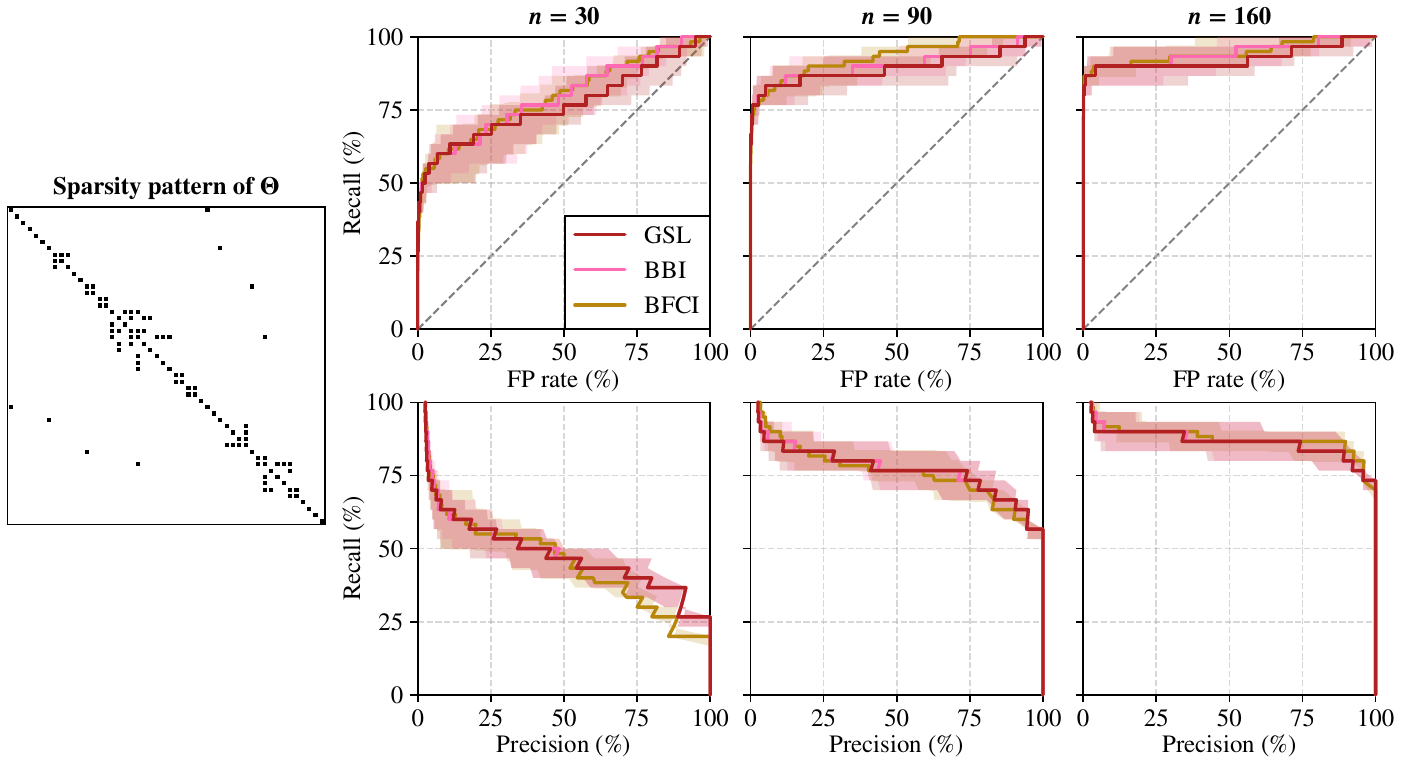}
\caption{ For the \texttt{494\_bus} example,
ROC curves (top row) 
and precision-recall curves (bottom row) for the graphs 
for the sequential growths 
induced by the GSL, BBI and BFCI selection rules with approximate full correction.
For each sample size, $100$ repetitions are performed for GSL and BBI, 
and $10$ for BFCI.  
The curves indicate the median accuracies, 
and the coloured regions the pointwise  interdecile ranges. See Section~\ref{sec:494bus}. }
\label{fig:LabelFig1}
\end{figure}

Next, we compare the accuracy of the graphs recovered via GSL 
with the ones obtained via Glasso and the Prec and PCorr naive sequential growth procedures. We consider the same sample sizes $n=30$, $90$ and $160$, and perform $100$ repetitions for each method.
The results are presented in Figure~\ref{fig:LabelFig2}, where we display both ROC and precision-recall curves.
We observe that in the sparse regime, the best accuracies are systematically achieved by GSL, 
Glasso catching up as the number of extracted edges increases. For $n=30$, the first crossing between the GSL and Glasso median curves corresponds to a graph with $40$ edges (recall that the number of true edges in this example is $30$); for $n=90$, the number of edges at crossing is $39$, and $36$ for $n=160$. 
The GSL variant of Algorithm~\ref{algo:RelaxSeqGrowth} also significantly outperforms Prec and PCorr in the early stages of the growth procedures. 

\begin{figure}[!htbp]
\centering
\includegraphics[width=0.75\linewidth]{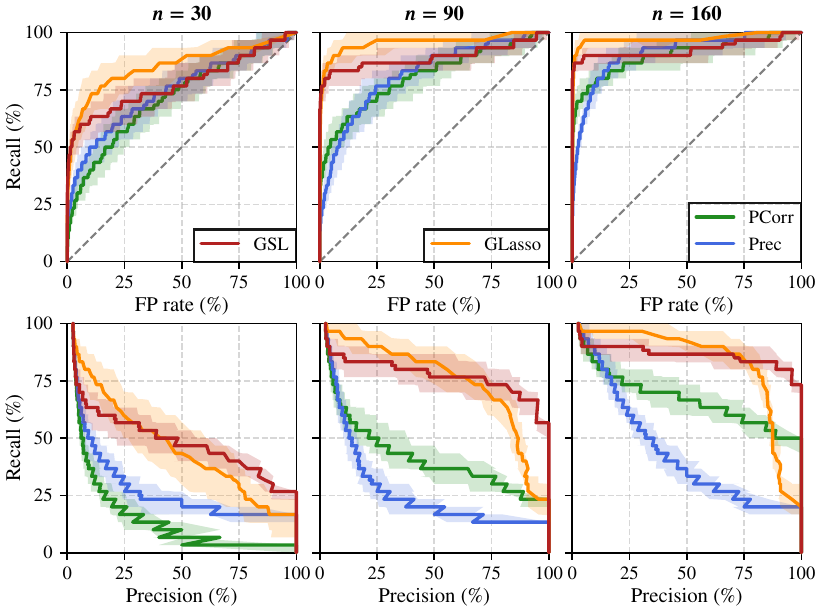}
\caption{
For the \texttt{494\_bus} example,
ROC curves (top row) 
and precision-recall curves (bottom row) for the graphs recovered via GSL, Glasso, Prec and Pcorr.
Sample sizes correspond to columns in the figure. 
The curves indicate the median accuracies over $100$ repetitions   
and the coloured regions the pointwise  interdecile ranges. 
 See Section~\ref{sec:494bus}}.   
\label{fig:LabelFig2}
\end{figure}

To further illustrate the ability of our approach to reliably extract sparse models, 
we use GSL to extract graphs with $k=30$ edges, which is the exact number of edges in the true underlying graph. 
We then assess the frequency of detection of the true edges, together with the distribution of the false positives (FPs).     
For comparison, we also extract graphs with $k=30$ edges using Prec and PCorr.   We set  $n=30$, and  perform $500$ repetitions. 
The results are presented in Figure~\ref{fig:LabelFig3}. 
As already observed, GSL significantly outperforms Prec and PCorr in the considered regime.  
Despite the relatively small sample size, 
certain true edges are detected with high frequency, 
supporting the ability of GSL to efficiently leverage the information geometry of the problem. 
Also, with the exception of a few structural outliers, 
the false-positive distribution is remarkably uniform.

\begin{figure}[!htbp]
\centering
\includegraphics[width=0.95\linewidth]{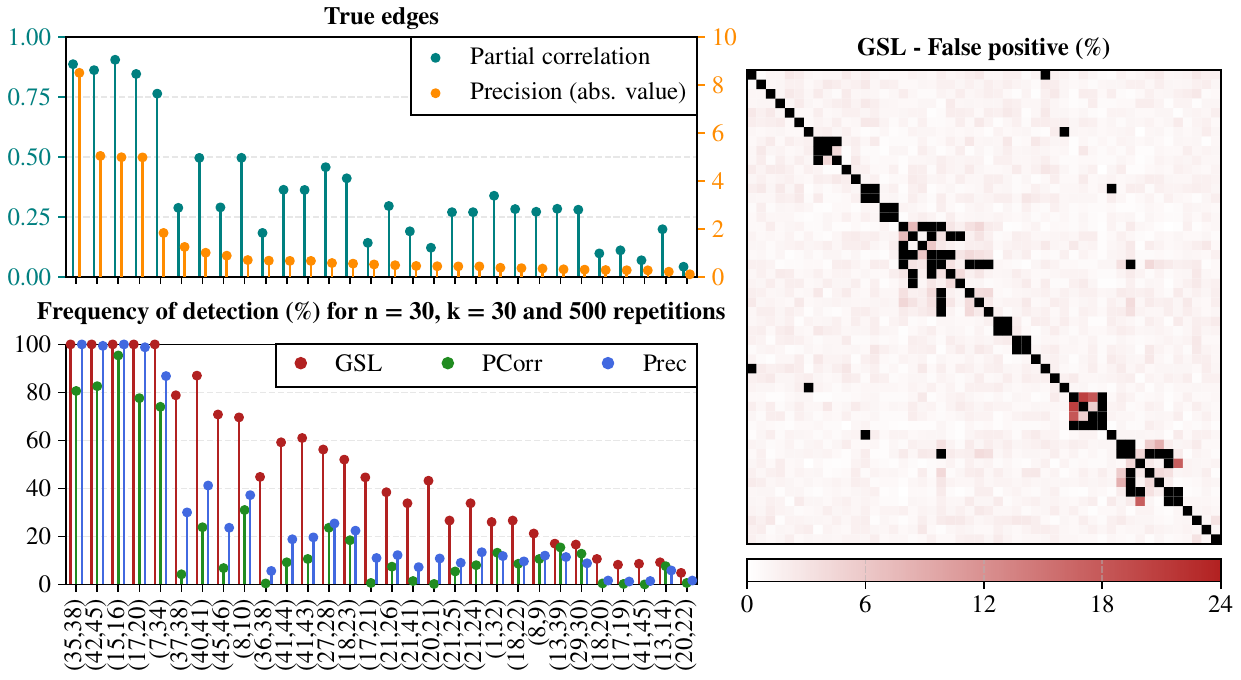}
\caption{Graphical representation of the true partial correlations and precisions of the $30$ true edges of the \texttt{494\_bus} example; 
the edges are ordered by decreasing absolute precision (top-left). For the sequential growth procedures GSL, Prec and PCorr, with $k=30$ 
(that is, $30$ edges are extracted), empirical frequency of detection of the true edges over $500$ repetitions; 
the sample size is $n=30$ (bottom-left). 
The empirical distribution of the false positives for GSL is also presented (right).  
See Section~\ref{sec:494bus}}.  
\label{fig:LabelFig3}
\end{figure}

In a follow up experiment, we assess the positions in which edges are activated during the full GSL growth procedure; we call this position the \emph{activation rank} of the respective edge.
We use $n=90$ and perform $500$ repetitions, and 
 order the edges according to their median activation rank.
Figure~\ref{fig:LabelFig4} displays the rank box plots of the $55$ foremost edges, together with the ones of all remaining true edges, of which there are only two. 
The results strongly support the ability of GSL to consistently identify relevant edges during the early stages of the growths. 
Apart from the edge $(20,22)$, whose signal is hardly distinguishable from noise at $n=90$ 
(see Figure~\ref{fig:LabelFig3}), 
we observe a clear discrepancy between the activation-rank distributions of the true and false positives.  
Note that this type of investigations relates to the notion \emph{stability selection} \cite{meinshausen2010stability,stars2010}; 
see Section~\ref{sec:Riboflavin} for a further illustration.  

\begin{figure}[!htbp]
\centering
\includegraphics[width=1.00\linewidth]{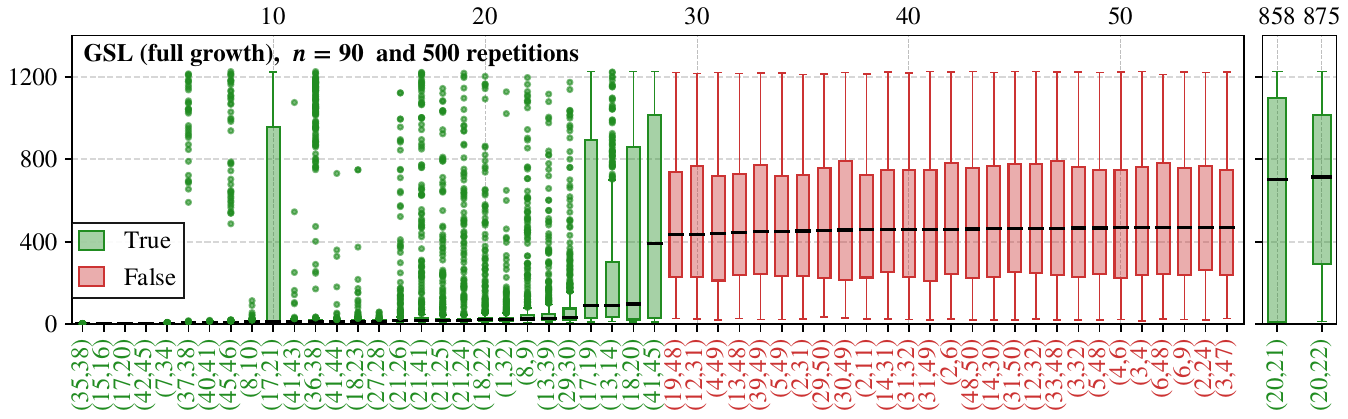}
\caption{
Distribution of the edge activation ranks
under the full GSL growth procedure ($k=1{,}225$)
for the \texttt{494\_bus} example with $n=90$ and $500$ repetitions. 
The edges are ordered according to their median activation position (top horizontal axis). 
The $55$ foremost edges are displayed (left), 
together with the remaining true edges (right). 
See Section~\ref{sec:494bus}.  } 
\label{fig:LabelFig4}
\end{figure}

To conclude this section, we investigate the impact of the precision of the approximate full corrections on the trade-off between computation cost and accuracy of the extraction.
We consider the GSL variant of Algorithm~\ref{algo:RelaxSeqGrowth} and a sample of size $n=90$.  
As stopping parameters for the inner loop (see Remark~\ref{rem:StopAlgo2}), we use $\alpha=1$ and $\beta=10$, 
and consider three different values of $\tau$,  
namely $\tau=10^{-1}$, $10^{-3}$ and $10^{-5}$, corresponding to a low, medium and high precision, respectively. 
For each value of $\tau$, we record the evolution of the number of inner iterations in Algorithm~\ref{algo:RelaxSeqGrowth} 
as a function of the number of outer iterations. 
In parallel, we record the evolution of the recalls of the recovered graphs. 
The results are presented in Figure~\ref{fig:LabelFig5}. 
As expected, the number of inner iterations increases as $\tau$ decreases. 
Interestingly, similar graph-recovery performances are observed for the three considered values of $\tau$, 
supporting the assertion that in the early stages of the growth process, 
relatively loose approximations of the full corrections can be used without jeopardising the accuracy of the recovered graphs.

\begin{figure}[!htbp]
\centering
\includegraphics[width=0.95\linewidth]{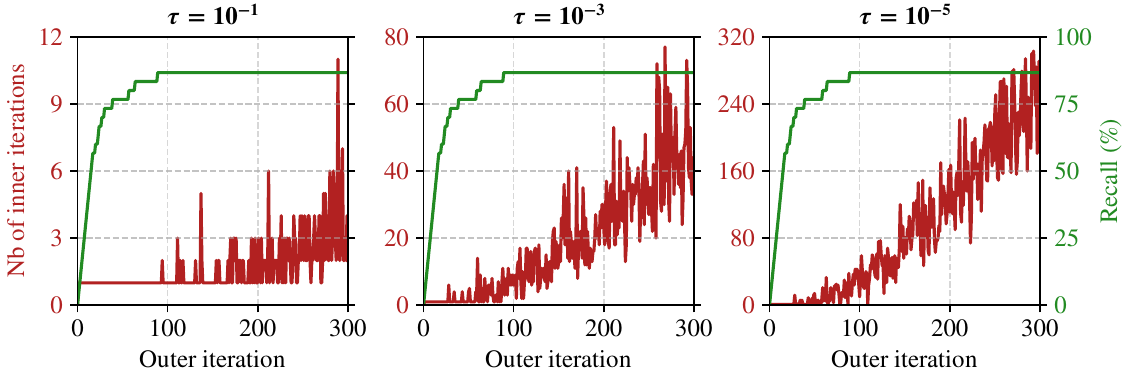}
\caption{For a sample ($n=90$) drawn in the setting of the \texttt{494\_bus} example, 
evolution of the number of inner iterations of the GSL variant of Algorithm~\ref{algo:RelaxSeqGrowth} 
as a function of the number of outer iterations, 
for three different values of the threshold parameter $\tau$, and with $\alpha=1$ and $\beta=10$ (see Remark~\ref{rem:StopAlgo2}). 
The evolution of the accuracy of the underlying graph-recovery process is also presented. 
 See Section~\ref{sec:494bus}.  } 
 \label{fig:LabelFig5}
\end{figure}

\subsection{Fully-synthetic simulation study}
\label{sec:FullSynthe}

We now assess the graph-recovery ability of the considered strategies across diverse scenarios. The considered graph structures are inspired by \cite{jacob2011sparse, mazumder2012graphical} and presented in Figure~\ref{fig:LabelFig6}.
They are encoded by a symmetric base matrix $\bfA\in\Sym^{d}$, whose non-zero upper-diagonal entries are independent realisations 
of a random variable $Z=RU$, where $R$ follows a Rademacher distribution, 
$U$ a uniform distribution over $[0.5,1.5]$, and with $R$ and $U$ independent; the diagonal of $\bfA$ is set to $0$. We next compute a solution $\bdd^{*}=(d_{i}^{*})_{i\in[d]}$
 to the semidefinite program
\begin{align*}
\text{minimise}\sum_{i\in[d]}d_{i}
\text{ subject to }\bfA+\diag(\bdd)\succcurlyeq 0 \text{ and }\bdd\in\BRnn^{d};
\end{align*}
we define $\bfB=\bfA+\diag(\bdd^{*})\in\SPSD^{d}$ and  set
$\bfM = \bfB + \eta \bfI_{d}$, $\eta>0$. 
We normalise the matrix $\bfM^{-1}$ (covariance to correlation), 
and use the resulting matrix $\Sigb$ as true covariance. 
The precision $\Theb=\Sigb^{-1}$ inherits the sparsity pattern from $\bfA$, and increasing $\eta$ reduces the magnitude of the underlying 
true precision and partial correlations (we may interpret $\eta$ as a signal-to-noise-ratio parameter). 
The considered graph structures are as follows.  
\begin{itemize}
\item {\textit{Random.}} The $m\in\BN$ non-zero upper-diagonal entries of $\bfA$ are selected uniformly at random over $\SU$. 
We consider two cases: $m=40$ and $m=200$. 
\item {\textit{Clique.}} The matrix $\bfA$ is populated such that the precision matrix $\Theb$ consists of $5$ diagonal blocks of order $10$.  
The number of true edges is $m=225$. 
\item {\textit{Hub.}} The matrix $\bfA$ is populated such that the resulting graph consists of $5$ distinct hubs, each involving $10$ nodes. The number of true edges is $m=45$. 
\end{itemize}
We use $d=50$ and consider $3$ sample sizes, namely $n=30$, $90$ and $160$, and in each case, 
we perform $100$ repetitions for Glasso, GSL, BBI, Prec and Pcorr, and $10$ for BFCI.

The results are presented in Figure~\ref{fig:LabelFig6} for $\eta=0.25$ (strong-signal case), 
and Figure~\ref{fig:LabelFig7} for $\eta=1$ (weak-signal case).  
They very much support the ability of the considered variants of Algorithm~\ref{algo:RelaxSeqGrowth} 
to reliably identify relevant edges while limiting the number of false detections. 
Depending on the sparsity of the extracted graphs, 
the best median performances are systematically achieved by either Glasso or GSL, BBI or BFCI. 
In line with the \texttt{494\_bus} example, the most accurate sparse graphs are often the ones produced by 
Algorithm~\ref{algo:RelaxSeqGrowth}, except for Clique, where Glasso generally achieves a better accuracy. 
Interestingly, in the Clique example, we observe a relatively significant difference 
between the three variants of Algorithm~\ref{algo:RelaxSeqGrowth}, 
the best performances being achieved by BFCI, followed by BBI and GSL. 
Overall, the differences in accuracy between Glasso and the variants of Algorithm~\ref{algo:RelaxSeqGrowth} 
are more significant for $\eta=0.25$ than for $\eta=1$. 
For the sparsest examples (that is, Random with $m=40$ and Hub) with $\eta=0.25$, 
the efficacy of the three variants of Algorithm~\ref{algo:RelaxSeqGrowth} is remarkable, 
especially for $n=90$ and $n=160$.  

\begin{figure}[!p]
\centering
\includegraphics[width=1.0\linewidth]{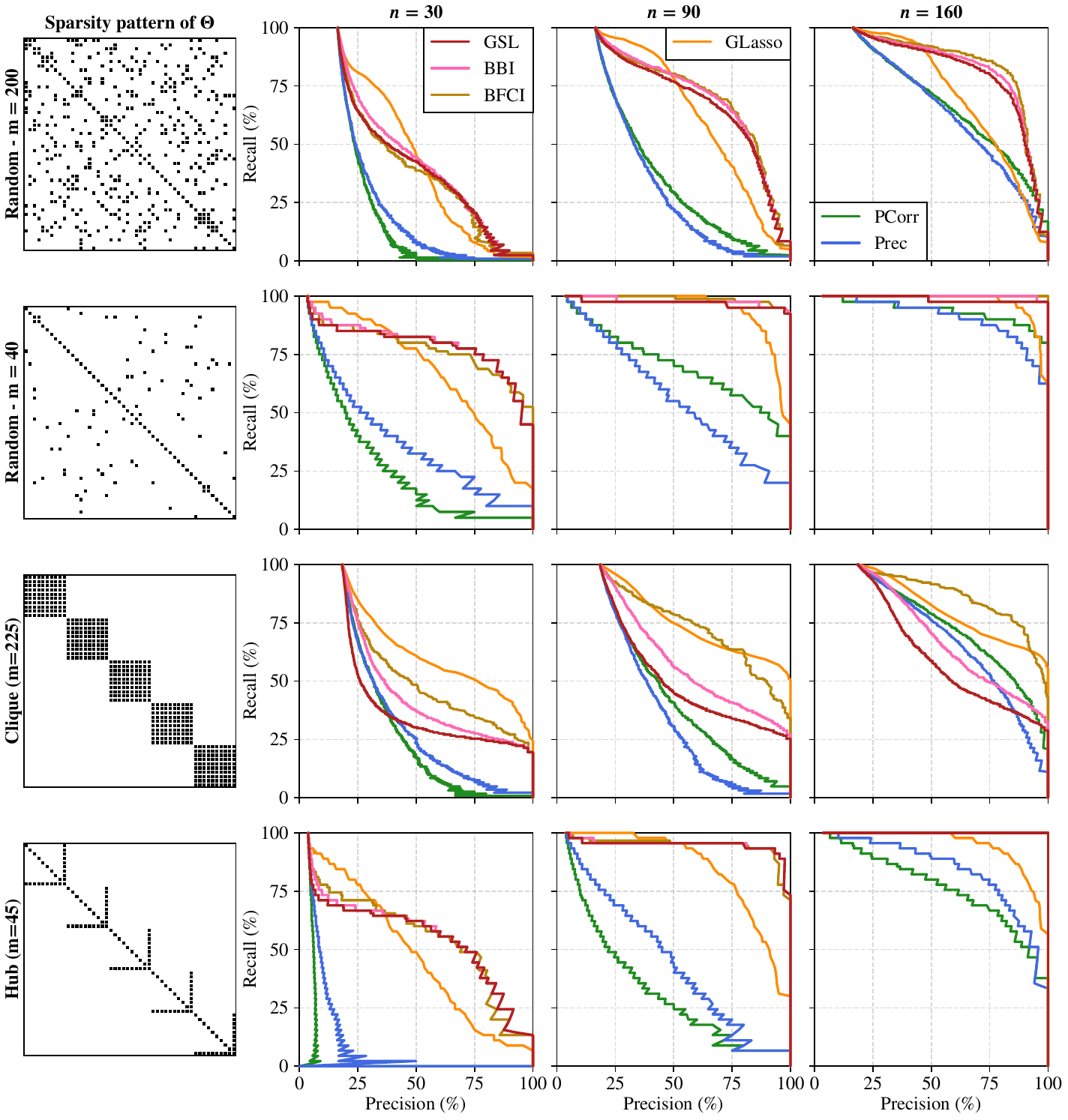}
\caption{For the fully-synthetic examples with $\eta=0.25$, 
precision-recall curves for the graphs recovered via GSL, BBI, BFCI, Glasso, Prec and Pcorr. 
The figure reads like a table, each row corresponding to a given example; 
for each example, three different sample sizes are considered. 
The curves indicate the median accuracies, 
over $10$ repetitions for BFCI, 
and $100$ repetitions for the other methods. 
See Section~\ref{sec:FullSynthe}. } 
 \label{fig:LabelFig6}
\end{figure}

\begin{figure}[!htb]
\centering
\includegraphics[width=1.0\linewidth]{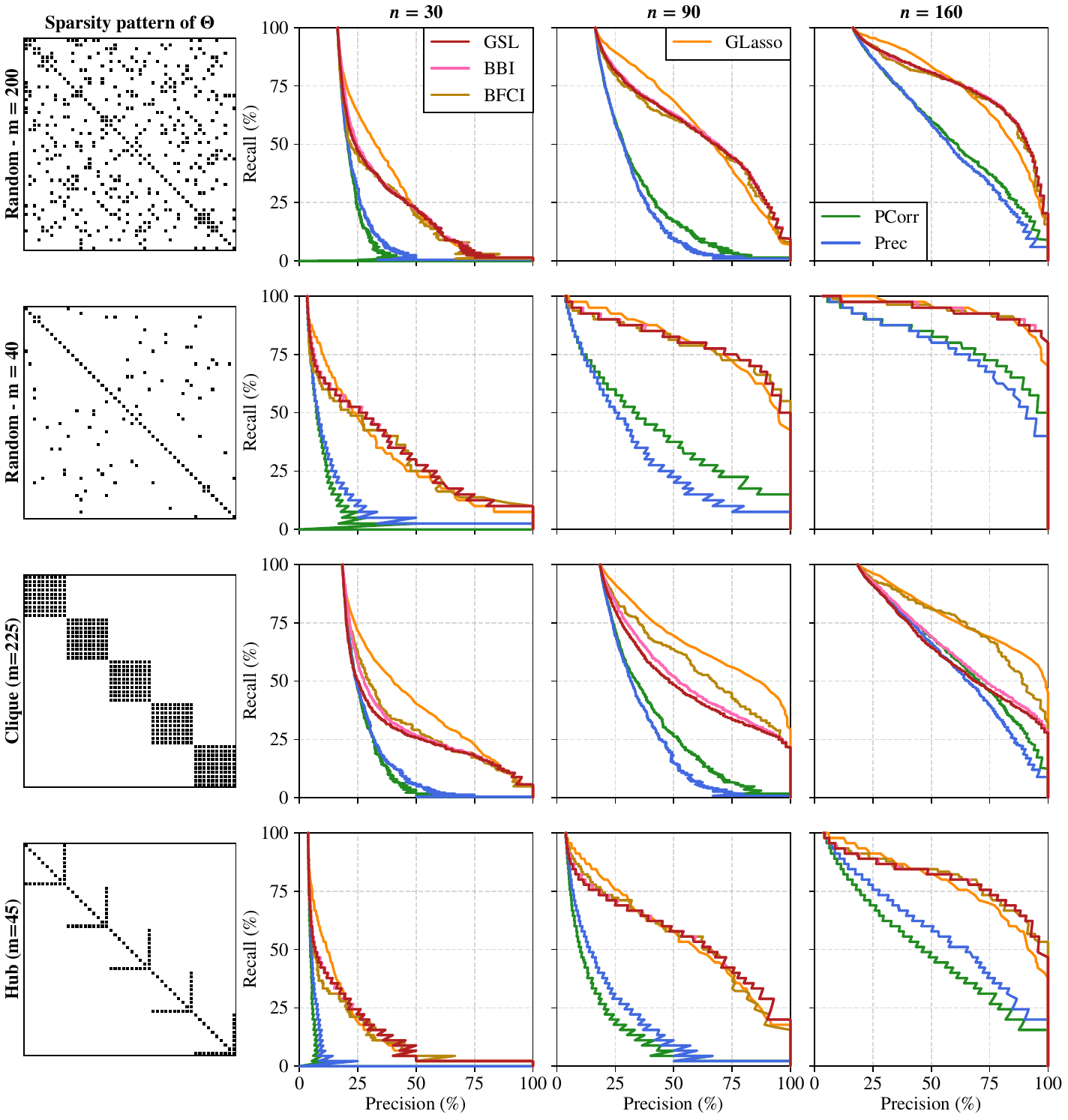}
\caption{Same setting as Figure~\ref{fig:LabelFig6}, but with $\eta=1$. } 
 \label{fig:LabelFig7}
\end{figure}

\subsection{Recovering a gene network: the Riboflavin data}
\label{sec:Riboflavin}

The Riboflavin data relate to the genome of bacterias involved in the production of the riboflavin vitamin. 
The data set consists of ${n=71}$ samples, reporting $4{,}088$ gene expressions; 
it is available in the R package \texttt{hdi} \cite{dezeure2015high}. 
Following \cite{buhlmann2014high, laszkiewicz2021thresholded},  
we only consider the $d=100$ genes with the highest empirical variances, 
and scale the data using the nonparanormal transformation \cite{liu2009nonparanormal}. 
We focus on the characterisation of sparse models (approximation of a non-necessarily sparse precision matrix), 
and more specially on the identification of gene clusters (gene-expression clustering). 
We use the GSL variant of Algorithm~\ref{algo:RelaxSeqGrowth}.

To gain insight into the range for which the generated graphs are likely to contain high proportions of true edges, 
we compute the edge-activation ranks induced by $500$ random subsamples, 
without replacement, 
of size $\lfloor n/2\rfloor=35$ 
(see \cite{meinshausen2010stability}),   
and order the edges according to their median activation position. 
In Figure~\ref{fig:LabelFig8}, 
we present the activation-rank box plots for the $300$ foremost edges. 
We do not observe 
a clear transition in the activation-rank distributions  
(such a transition is observed in Figure~\ref{fig:LabelFig4}), 
so that a visual inspection does not rule out that the underlying true graph could contain more than $300$ edges;  
observe however that the sample size is small, and that subsampling might induce artefacts. 
Sparse graphs containing up to approximately $100$ edges appear likely to contain significant proportions of true edges.

\begin{figure}[!htb]
\centering
\includegraphics[width=1.00\linewidth]{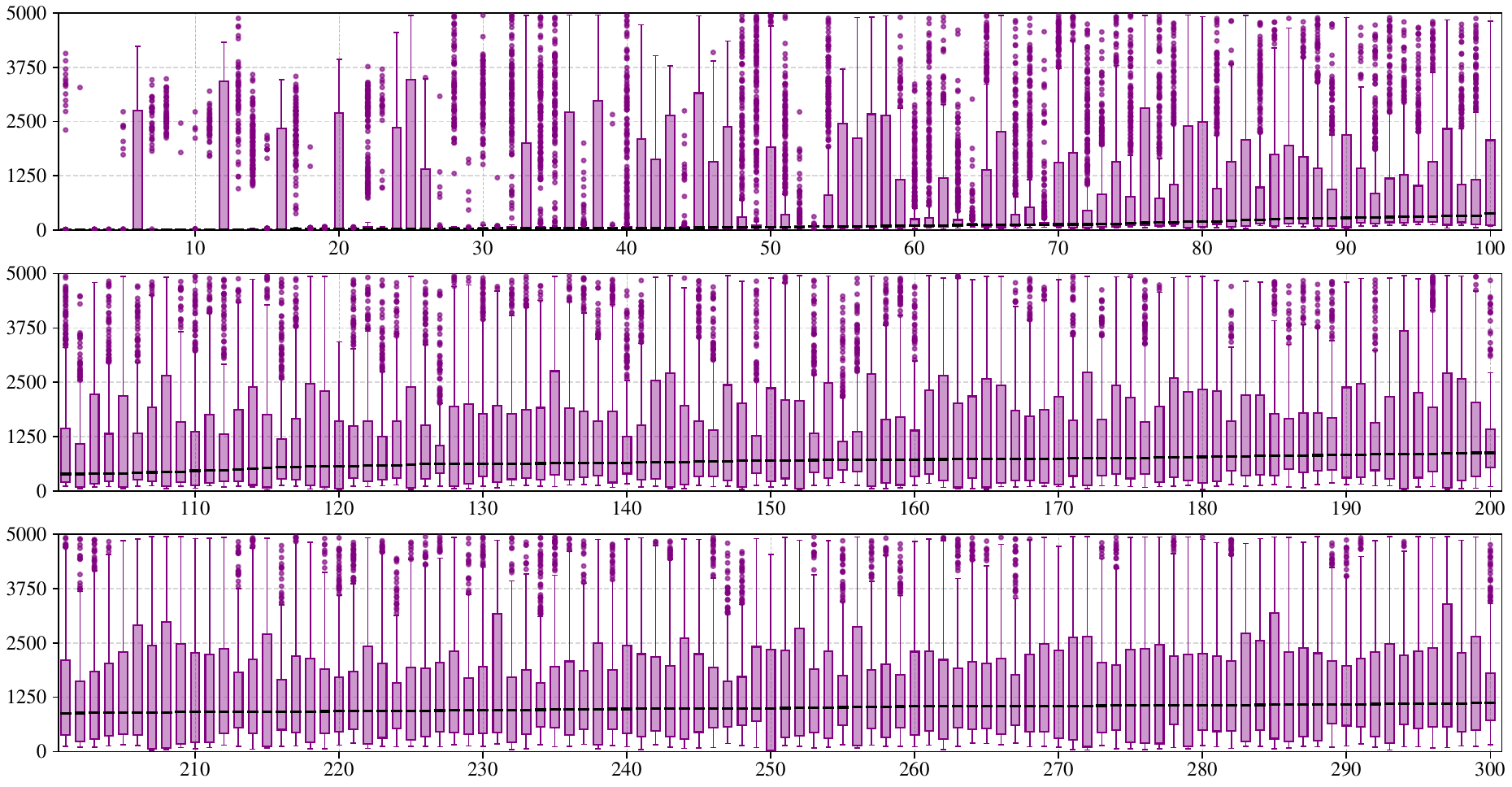}
\caption{For the Riboflavin data and $500$ random subsamples of size $\lfloor n/2\rfloor=35$ (without replacement),
distribution of the edge-activation ranks
under the full GSL growth procedure ($k=4{,}950$). 
The edges are ordered according to their median activation position; 
the $300$ foremost edges are displayed. 
See Section~\ref{sec:Riboflavin}.  } 
 \label{fig:LabelFig8}
\end{figure}

In Figure~\ref{fig:LabelFig9}, we present the graphs with $k=40$, $80$ and $120$ edges
defined by the subsampling-based activation ranks displayed in Figure~\ref{fig:LabelFig8};
in each case, the edge set consists of the $k$ foremost edges in terms of median activation rank. 
For comparison, we also present in Figure~\ref{fig:LabelFig10}, for the same values of $k$, the graphs produced by the GSL growth procedure applied to the full data set.
The graphs reveal some intriguing patterns,  
illustrating the potential interest of the considered growth procedures for gene-expression clustering.
Notably, the number of clusters can be easily adjusted by varying the edge-number parameter $k$. 
The graphs obtained for $k=80$ share similarities with the one presented in \cite{laszkiewicz2021thresholded}, 
while the ones obtained for $k=120$ appear relatively close to the one produced in \cite{buhlmann2014high}. 
Note that the graphs in Figure~\ref{fig:LabelFig9} require significantly more compute than their counterparts in Figure~\ref{fig:LabelFig10}. More precisely, Figure~\ref{fig:LabelFig9} involves $500$ full growths, against one partial growth for Figure~\ref{fig:LabelFig10}.

\begin{figure}[!htb]
\centering
\includegraphics[width=1.00\linewidth]{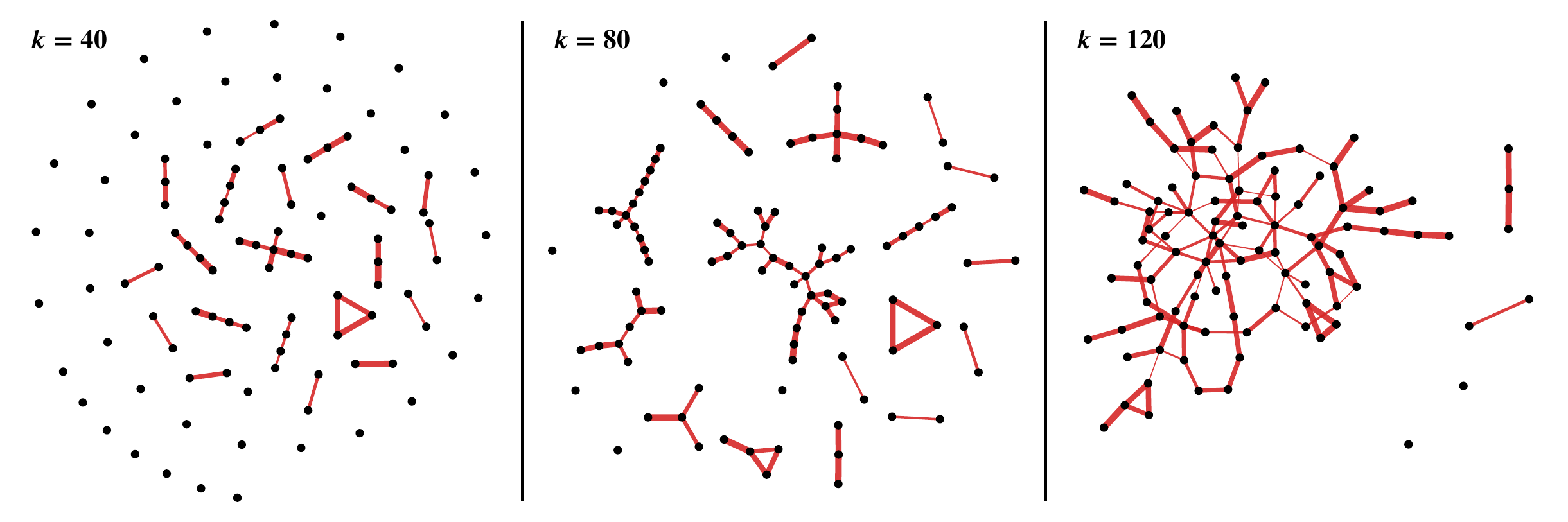}
\caption{ For the Riboflavin data,
graphical representation of three graphs defined by 
subsampling-based activation ranks displayed in Figure~\ref{fig:LabelFig8}.  
For a given value of $k$, the edge set of the graph consists of the $k$ foremost edges 
in Figure~\ref{fig:LabelFig8}. 
The width of the edges relates to their median rank, 
wider lines indicating lower median ranks.  
See Section~\ref{sec:Riboflavin}. } 
 \label{fig:LabelFig9}
\end{figure}

\begin{figure}[!htb]
\centering
\includegraphics[width=1.00\linewidth]{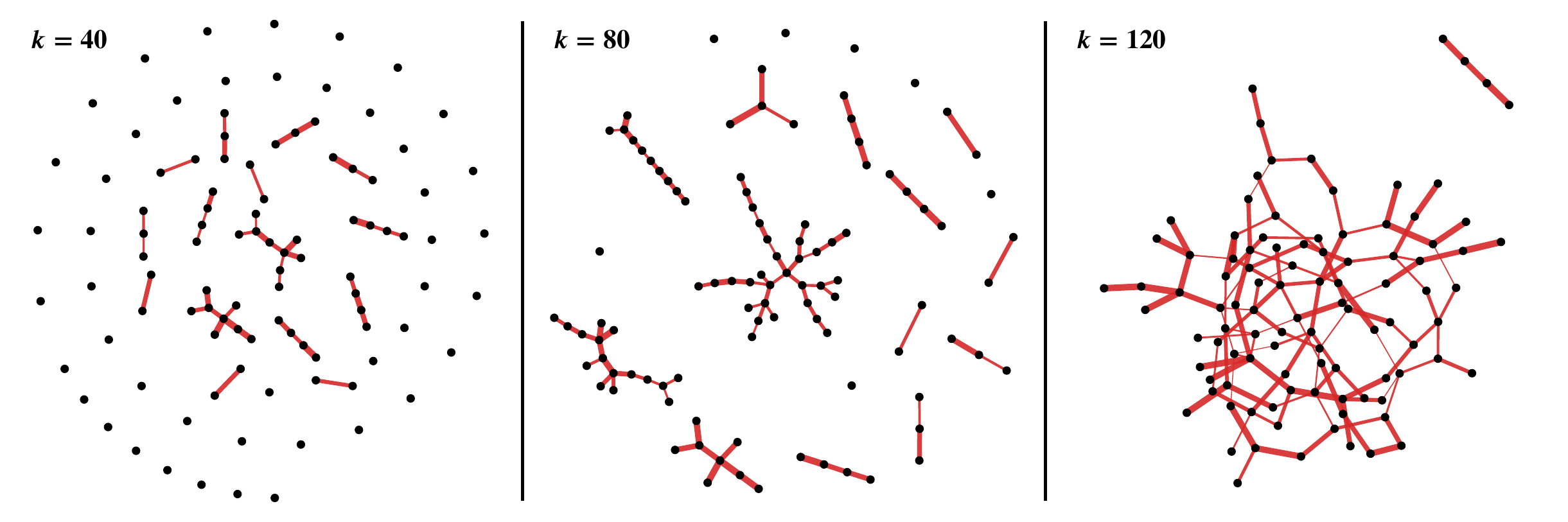}
\caption{For the Riboflavin data,
graphical representation of three graphs extracted by the GSL sequential growth procedure applied to the full data set.  
The width of the edges relates to their position in the growth process, 
wider lines indicating earlier activations.  
See Section~\ref{sec:Riboflavin}}. 
 \label{fig:LabelFig10}
\end{figure}

\section{Concluding discussion}
\label{sec:Conclusion}

We presented a class of regularisation-free 
approaches for Gaussian graphical inference 
based on the information-geometry-driven sequential growth of initially edgeless graphs.      
Leveraging the analogy between coordinate update and edge activation,        
we characterised the fully-corrective descents corresponding to information-optimal growths, 
and proposed numerically efficient strategies for their approximation. 
Our numerical experiments highlight the ability of the discussed procedures 
to reliably extract sparse models while limiting the number of false detections,  
providing an interesting complement to the existing methodologies.  
Notably, the sequential nature of the discussed procedures offers an original diagnostic tool 
where the activation rank of the edges acts as a measure of their informational relevance, 
allowing for a nuanced assessment of the detected dependence structures. 
Aligning with the idea of stability selection,  
we also illustrated how the considered approaches can be combined with subsampling techniques 
to support the identification of relevant graphs. 

Our investigations mostly focused on the graph-recovery ability of the proposed procedures,    
and follow-up studies could assess their relevance for parameter learning. 
The implementations of sequential growth procedures in constrained settings
(such as total positivity, see for instance \cite{lauritzen2022locally, lauritzen2019maximum}) 
could also offer an interesting avenue for future research. 
Finally, in the stability-selection framework, 
the design of relevant tests on the edge-activation ranks,  
together with the implementation of censoring techniques to account for partial growths,   
could help further improve the deployability and scalability of sequential-growth-based methodologies.

\appendix

\section{Proofs and technical results}
\label{sec:ProofsAndCo}

This section gathers proofs of the results stated in the main body of the paper. 
It also contains a technical result (Lemma~\ref{lem:ImprovLSmooth}) used to prove the convergence of Algorithm~\ref{algo:GraphOptMatApp12Up}. 
Recall that 
\begin{align}\label{eq:TraceProdIneq}
\lambda_{\min}(\bfQ)\trace(\bfM)\leqslant \trace(\bfQ\bfM)\leqslant \lambda_{\max}(\bfQ)\trace(\bfM),
\end{align}
for all $\bfQ\in \Sym^{d}$ and $\bfM\in\SPSD^{d}$.

\vspace{0.5\baselineskip}
\noindent\textit{Proof of Lemma~\ref{lem:fSStrictCvx}}. 
Set $\bfQ^{-1}=\bfA^{2}$, with $\bfA\in\SPD^{d}$; 
the strict convexity of $f_{\bfS}$ follows by observing that 
$D_{\bfQ}^{2} f_{\bfS}(\bfH,\bfH) = \|\bfA\bfH\bfA\|_{\Frob}^{2} \geqslant 0$, $\bfH\in\Sym^{d}$, 
with equality if and only if $\bfH=0$. 

Assume that $\bfS$ is invertible. We have $D_{\bfQ} f_{\bfS}=0$ if and only if $\bfQ=\bfS^{-1}$, 
characterising the minimum of $f_{\bfS}$. 
Consider a sequence $(\bfQ_{k})_{k\in\BN}\subset\SPD^{d}$. 
If $\lambda_{\max}(\bfQ_{k})\to+\infty$, then from \eqref{eq:TraceProdIneq}, we have
\begin{align*}
f_{\bfS}(\bfQ_{k})\geqslant
\lambda_{\min}(\bfS)\lambda_{\max}(\bfQ_{k}) - d\log(\lambda_{\max}(\bfQ_{k})) \to +\infty.      
\end{align*}
 Also, as $\trace(\bfS\bfQ_{k})\geqslant 0$, $k\in\BN$,
 we obtain that if $\det(\bfQ_{k})\to0$, then $f_{\bfS}(\bfQ_{k})\to+\infty$, 
 proving the coercivity of $f_{\bfS}$ over $\SPD^{d}$.  

Assume that $\bfS$ is singular.  
Setting $\bfS=\diag(\lambda_{1},\cdots,\lambda_{p}, 0, \cdots, 0)$, 
with $p<d$ and $\lambda_{i}>0$,   
and $\bfQ_{k}=\diag(1/\lambda_{1},\cdots,1/\lambda_{p}, k, \cdots, k)$, $k\in\BN$, 
we obtain 
\begin{align*}
f_{\bfS}(\bfQ_{k})
= p + \sum_{i=1}^{p}\log(\lambda_{i}) - (d-p)\log(k)
\to-\infty \quad \text{as $k\to+\infty$. }
\end{align*}
A similar argument applies to a general matrix $\bfS$ via diagonalisation. 
\hfill\qed

\vspace{0.5\baselineskip}
\noindent\textit{Proof of Lemma~\ref{lem:StongConvexity}.}
We have 
$\nabla f_{\bfS}(\bfQ_{1})-\nabla f_{\bfS}(\bfQ_{2})
=\bfQ_{2}^{-1}-\bfQ_{1}^{-1}=\bfQ_{1}^{-1}(\bfQ_{1}-\bfQ_{2})\bfQ_{2}^{-1}$. 
Both inequalities then follow from \eqref{eq:TraceProdIneq}. 
\hfill\qed

\vspace{0.5\baselineskip}
\noindent\textit{Proof of Lemma~\ref{lem:GraphLikeliExist}.}
The existence of a minimum follows from 
Lemma~\ref{lem:fSStrictCvx}. 
We indeed have
\begin{align*} 
\inf\big\{f_{\bfS}(\bfQ) \,\big|\, \edge(\bfQ)\subseteq\SE \big\}
=\inf\big\{f_{\bfS}(\bfQ) \,\big|\, \edge(\bfQ)\subseteq\SE 
\text{ and } f_{\bfS}(\bfQ)\leqslant f_{\bfS}( \bfQ_{\bfS,\emptyset} ) \big\}, 
\end{align*}
and $\CK=\big\{ \bfQ\in\SPD^{d} \,\big|\, \edge(\bfQ)\subseteq\SE 
\text{ and } f_{\bfS}(\bfQ)\leqslant f_{\bfS}( \bfQ_{\bfS,\emptyset} ) \big\}$
is convex and compact (intersection of a convex compact sublevel set of $f_{\bfS}$ with a linear subspace of $\Sym^{d}$) and non-empty as $\bfQ_{\bfS,\emptyset}$ is an element. 
Strict convexity of $f_{\bfS}$ ensures uniqueness.  
The zeroing of the entries 
of the gradient is a consequence of the openness of
$\{\bfQ\in\SPD^{d} \,\big|\, \edge(\bfQ)\subseteq\SE\}$ in $\vspan\{\bfB(i,j) \,|\, (i,j)\in\SD\cup\SE\}$.  
\hfill\qed

\vspace{0.5\baselineskip}
\noindent\textit{Proof of Lemma~\ref{lem:BlockUpdate}.} 
We have $\tilde{\bfQ}\in(\bfQ + \vspan\{\bfB(i,j) \,|\, \text{$i, j\in I$, $i\leqslant j$}\})$, and a direct computation gives 
$\tilde{\bfR}^{-1}=\tilde{\bfQ}$. 
Moreover, we have $\tilde{\bfQ}\succ 0$ (that is, $\tilde{\bfQ}$ is symmetric positive-definite), 
since $\tilde{\bfQ}_{I^{c},I^{c}}=\bfQ_{I^{c},I^{c}}\succ 0$ 
and $\tilde{\bfQ}_{I,I}-\tilde{\bfQ}_{I,I^{c}}(\tilde{\bfQ}_{I^{c},I^{c}})^{-1}\tilde{\bfQ}_{I^{c},I}=(\bfS_{I,I})^{-1}\succ 0$,
see e.g. \cite[Theorem~7.7.7]{horn2013matrix}. 
As $[\nabla f_{\bfS}(\tilde{\bfQ})]_{I,I}=0$, the optimality follows. 
A direct computation provides the value of the improvement $f_{\bfS}(\bfQ) - f_{\bfS}(\tilde{\bfQ})$. 
\hfill\qed

\begin{lemma}\label{lem:ImprovLSmooth}
Consider $\bfS$ and $\bfQ\in\SPD^{d}$,  
and $\bfH\in\Sym^{d}\backslash\{0\}$.  
For $\alpha\in\BR$, set $\bfQ_{\alpha}=\bfQ+\alpha\bfH$. 
There exists $\alpha^{*}$ such that $\mip[]{\nabla f_{\bfS}(\bfQ_{\alpha^{*}}), \bfH}_{\Frob}=0$. 
Moreover, for any convex set $\CK\subset\SPD^{d}$ such that $\bfQ$ and $\bfQ_{\alpha^{*}}\in\CK$ 
and $\lambda_{\min}(\CK)>0$,
we have $|\alpha^{*}|\geqslant |\mip[]{\nabla f_{\bfS}(\bfQ), \bfH}_{\Frob}|/(L\|\bfH\|_{\Frob}^{2})$ and 
$f_{\bfS}(\bfQ) - f_{\bfS}(\bfQ_{\alpha^{*}})
\geqslant \mip[]{\nabla f_{\bfS}(\bfQ), \bfH}_{\Frob}^{2}/(2L\|\bfH\|_{\Frob}^{2})$, 
with $L=1/\lambda_{\min}^{2}(\CK)$. 
\end{lemma}

\begin{proof}
The existence of $\alpha^{*}$ follows by convexity and coercivity of $f_{\bfS}$ (Lemma~\ref{lem:fSStrictCvx}); 
observe that $\bfQ_{\alpha^{*}}$ is the unique minimum of $f_{\bfS}$ over the convex set 
$\{\bfQ_{\alpha} \,|\, \alpha\in\BR\}\cap\SPD^{d}$ 
(exact line search from $\bfQ$ along $\bfH$). 
By Cauchy-Schwartz
and from the $L$-smoothness of $f_{\bfS}$ over $\CK$ (Lemma~\ref{lem:StongConvexity}), 
we get 
\begin{align*}
|\mip[]{\nabla f_{\bfS}(\bfQ), \bfH}_{\Frob}|
&= |\mip[]{\nabla f_{\bfS}(\bfQ)-\nabla f_{\bfS}(\bfQ_{\alpha^{*}}), \bfH}_{\Frob}|\\
&\leqslant \|\nabla f_{\bfS}(\bfQ)-\nabla f_{\bfS}(\bfQ_{\alpha^{*}})\|_{\Frob}\|\bfH\|_{\Frob}
\leqslant |\alpha^{*}|L\|\bfH\|_{\Frob}^{2}, 
\end{align*}
providing the expected lower bound for $|\alpha^{*}|$. 

From the $L$-smoothness of $f_{\bfS}$, 
for all $\alpha$ such that $\bfQ_{\alpha}\in\CK$, 
we have 
\begin{align*}
f_{\bfS}(\bfQ_{\alpha})
\leqslant f_{\bfS}(\bfQ) + \alpha\mip[]{\nabla f_{\bfS}(\bfQ), \bfH}_{\Frob}
+\alpha^{2}L\|\bfH\|_{\Frob}^{2}/2
=\varphi(\alpha).  
\end{align*}
The minimum of the function $\varphi$ over $\BR$ is reached at 
$\bar{\alpha}=-\mip[]{\nabla f_{\bfS}(\bfQ), \bfH}_{\Frob}/(L\|\bfH\|_{\Frob}^{2})$. 
We thus have $|\alpha^{*}| \geqslant |\bar{\alpha}|$, and $\alpha^{*}$ and $\bar{\alpha}$ have same sign;   
the convexity of $\CK$ entails $\bfQ_{\bar{\alpha}}\in\CK$. 
We obtain
\begin{align*}
f_{\bfS}(\bfQ_{\alpha^{*}})
\leqslant f_{\bfS}(\bfQ_{\bar{\alpha}})
\leqslant \varphi(\bar{\alpha})
=f_{\bfS}(\bfQ) - \mip[]{\nabla f_{\bfS}(\bfQ), \bfH}_{\Frob}^{2}/(2L\|\bfH\|_{\Frob}^{2}),   
\end{align*}
completing the proof. 
\end{proof}

\noindent\textit{Proof of Theorem~\ref{thm:ConvergenceAlgo2GS}. } 
Consider $t\in\BN$ and set $\bfQ=\bfQ^{(t-1)}$ and $\tilde{\bfQ}=\bfQ^{(t)}$.  
Also, denote by $\check{\bfQ}$ the iterate of an exact line search from $\bfQ^{(t-1)}$ along $\bfB=\bfB(i,j)^{(t)}$. 
By definition of  $\tilde{\bfQ}$ and $\check{\bfQ}$, 
and from Lemma~\ref{lem:ImprovLSmooth} with $\CK=\CK_{0,\SE}$, we have 
\begin{align*}
f_{\bfS}(\bfQ)-f_{\bfS}(\tilde{\bfQ})
\geqslant f_{\bfS}(\bfQ)-f_{\bfS}(\check{\bfQ})
\geqslant \mip[]{\nabla f_{\bfS}(\bfQ), \bfB}_{\Frob}^{2}/(2L). 
\end{align*}  
For $\SE\subseteq\SU$, define the semi-norm 
$\|\bfM\|_{\Frob,\SE}^{2}=\sum_{(i,j)\in\SD\cup\SE}\mip[]{\bfM, {\bfB(i,j)}}_{\Frob}^{2}$, 
$\bfM\in\Sym^{d}$,
and denote the underlying semi-inner product accordingly.
By definition of the GS rule, 
we have $\mip[]{\nabla f_{\bfS}(\bfQ), \bfB}_{\Frob}^{2}
\geqslant \|\nabla f_{\bfS}(\bfQ)\|_{\Frob,\SE}^{2}/m$,  
and so   
\begin{align}\label{eq:IneqLsmoothX1}
f_{\bfS}(\bfQ)-f_{\bfS}(\tilde{\bfQ})
\geqslant \|\nabla f_{\bfS}(\bfQ)\|_{\Frob,\SE}^{2}/(2Lm). 
\end{align}  

From the $\mu$-strong convexity of $f_{\bfS}$ over $\CK_{0,\SE}$ and
the existence of $\bfQ_{\bfS,\SE}$ (Lemmas~\ref{lem:StongConvexity} and \ref{lem:GraphLikeliExist}), 
observing that $\bfQ_{\bfS,\SE}$ and $\bfQ$ are both supported by $G=([d],\SE)$,  
we have 
\begin{align*}
f_{\bfS}(\bfQ_{\bfS,\SE})
 \geqslant f_{\bfS}(\bfQ) + \mip[]{\nabla f_{\bfS}(\bfQ), \bfQ_{\bfS,\SE}-\bfQ}_{\Frob,\SE}
+\mu\|\bfQ_{\bfS,\SE}-\bfQ\|_{\Frob,\SE}^{2}/2
=\psi(\bfQ_{\bfS,\SE}).  
\end{align*}
The map $\bfM\mapsto\psi(\bfM)$, $\bfM\in\Sym^{d}$, is convex and minimal at 
$\bar{\bfM}=\bfQ-\nabla f_{\bfS}(\bfQ)/\mu$, and so  
\begin{align}\label{eq:IneqMuStrongX11}
f_{\bfS}(\bfQ_{\bfS,\SE})
\geqslant \psi(\bar{\bfM})=f_{\bfS}(\bfQ) - \|\nabla f_{\bfS}(\bfQ)\|_{\Frob,\SE}^{2}/(2\mu).  
\end{align}
Combining \eqref{eq:IneqLsmoothX1} and \eqref{eq:IneqMuStrongX11}, we get 
\begin{align}\label{eq:GapBoundImprov}
f_{\bfS}(\bfQ^{(t-1)}) - f_{\bfS}(\bfQ^{(t)}) 
\geqslant \frac{\mu}{mL}\big(f_{\bfS}(\bfQ^{(t-1)}) - f_{\bfS}(\bfQ_{\bfS,\SE})\big), \quad 
t\in\BN,  
\end{align}
and the result follows. 
\hfill\qed

\paragraph{Acknowledgements.}
 The first author thankfully acknowledges financial support from the DTP 2224 Cardiff University (UKRI grant EP/W524682/1)  at the School of Mathematics (project reference 2926088).


{\small
\bibsep0.93mm

}

\end{document}